\documentclass[journal]{IEEEtran}
\usepackage{helvet}         % selects Helvetica as sans-serif font
\usepackage{courier}        % selects Courier as typewriter font
\usepackage{type1cm}        % activate if the above 3 fonts are
\usepackage{makeidx}         % allows index generation
\usepackage{graphicx}        % standard LaTeX graphics tool

\usepackage{multicol}        % used for the two-column index
\usepackage[bottom]{footmisc}% places footnotes at page bottom
\usepackage{ifthen}
\usepackage{subfigure}
\usepackage[usenames,dvipsnames]{color}
\usepackage{cite}

\makeindex            
\usepackage{graphics} % for pdf, bitmapped graphics files
\usepackage{graphicx} % for pdf, bitmapped graphics files
\usepackage{epsfig} % for postscript graphics files

\usepackage{times} % assumes new font selection scheme installed
\usepackage{mathtools}  % loads amsmath
\usepackage{amssymb}  % assumes amsmath package installed
\usepackage{amsthm}
\usepackage{amsfonts}
\usepackage{dsfont}
\usepackage{mathrsfs}
\DeclareGraphicsExtensions{.pdf,.png,.jpg,.eps}

\usepackage{subfig}
\usepackage{verbatim}
\usepackage{comment}
\usepackage[ruled,linesnumbered]{algorithm2e}
\usepackage{multirow}
\usepackage{enumitem}
\usepackage{caption}
\usepackage{setspace}
\DeclareGraphicsExtensions{.pdf,.png,.jpg,.eps}

\usepackage{booktabs}
\usepackage[percent]{overpic}
\newcommand{\rr}{{\mathbb{R}}}
\newcommand{\T}{{\text{T}}}
\DeclareMathOperator{\col}{col}

\newboolean{showcomments}
\setboolean{showcomments}{true}
\newcommand{\wang}[1]{\ifthenelse{\boolean{showcomments}}
	{ \textcolor{red}{(ZW:  #1)}}{}}
\newcommand{\fliu}[1]{\ifthenelse{\boolean{showcomments}}
	{ \textcolor{red}{(FL:  #1)}}{}}
\newcommand{\peng}[1]{\ifthenelse{\boolean{showcomments}}
	{ \textcolor{red}{(PY:  #1)}}{}}

\theoremstyle{definition}
\newtheorem{theorem}{Theorem}

\newtheorem{proposition}[theorem]{Proposition}

\theoremstyle{definition}

\newtheorem{remark}{Remark}

\newtheorem{component}{Typical Device}

\begin{document}
	% \setstretch{0.99}	
	
	\title{{Compositional} and Equilibrium-Free Stability Certification for Power Systems--Part II: Algorithms and Applications}
	
	\author{Peng~Yang, \IEEEmembership{Member,~IEEE}, 
		~Yifan~Su,
		~Xiaoyu~Peng,
		~Hua~Geng, \IEEEmembership{Fellow,~IEEE},
		~Feng~Liu, \IEEEmembership{Senior~Member,~IEEE}
		\thanks{This work was supported by the National Natural Science Foundation of China (No. 52407137) and the Young Talent Fund of Association for Science and Technology in Shaanxi, China (No. 20250431).}
		\thanks{P. Yang is with the School of Electronics and Information, Xi’an Polytechnic University, Xi'an, 710048, China. (e-mail: p-yang13@tsinghua.org.cn)}
			\thanks{Y. Su is with the School of Electrical Engineering, Chongqing University, Chongqing, 400044, China. (e-mail: suyf@cqu.edu.cn)}
			\thanks{X. Peng is with the Department of Electrical Engineering, Tsinghua University, Beijing, 100084, China. (e-mail: pengxy19@tsinghua.org.cn)}			
			\thanks{H. Geng is with the Department of Automation, Tsinghua University, Beijing, 100084, China. (e-mail: genghua@tsinghua.edu.cn)}			
			\thanks{F. Liu is with the Department of Electrical Engineering, Tsinghua University, Beijing, 100084, China. (e-mail: lfeng@tsinghua.edu.cn) (Corresponding author)}
		
		%		\thanks{This work was supported by the National Natural Science Foundation of China (No. 51677100, U1766206) (Corresponding author: Feng Liu)}
		%		\thanks{P. Yang and F. Liu are with the Department of Electrical Engineering, Tsinghua University, Beijing, 100084, China. (e-mail: lfeng@tsinghua.edu.cn)}	
	}
	
	\maketitle
	
	\begin{abstract}
        This two-part study develops a {compositional} and equilibrium-free framework for power system stability analysis. Building on the theoretical foundation established in Part I \cite{YP_part-I}, i.e., local delta dissipativity (LDD), this second part translates the theory into a practical assessment tool. We address two core implementation challenges: 1) how to certify the LDD conditions for heterogeneous power device models, and 2) how to verify the network-wide coupling condition in a scalable manner. To this end, we present a method that employs Krasovskii-type storage functions to verify LDD conditions for heterogeneous power devices and applies to a wide range of nonlinear power device models. For coupling conditions, we propose an ADMM-based distributed algorithm, enhanced with a $p$-check subroutine, for interconnection verification. Our method enables three primary application scenarios that have been hindered by equilibrium-point-oriented and centralized methods: stability assessment for systems with multiple equilibria, rapid stability screening under varying operating conditions, and a privacy-preserving distributed assessment architecture for large-scale systems. Case studies on modified IEEE 9-bus, 39-bus, and 118-bus test systems demonstrate the effectiveness of the proposed methods across different network configurations and scales.
	\end{abstract}
	
	\begin{IEEEkeywords}
		Power system stability; local delta dissipativity; {compositional stability certification}; {equilibrium-free analysis}; ADMM.
	\end{IEEEkeywords}

	\IEEEpeerreviewmaketitle

	\section{Introduction}
	\label{sec:1}
	\subsection{Background and Motivation}
	\IEEEPARstart{S}{tability} analysis of power systems has become increasingly challenging with the growing integration of renewable energy, the proliferation of heterogeneous devices, and the complexity of modern grids \cite{huang2021decentralized,mevsanovic2021scalable}. {Recent studies have addressed various aspects of stable and secure operation, including load-frequency control under sampling and communication constraints \cite{yang2026nonuniform,tang2024multiple,yang2026aperiodic} and security issues in sampled-data control systems \cite{chen2025asynchronous}.} {Focusing on the stability of interconnected heterogeneous devices, Part I of this two-part paper introduced local delta dissipativity (LDD) and developed a compositional, equilibrium-free stability theory based on LDD conditions \cite{YP_part-I}.} This theory provides a foundation for overcoming limitations of traditional centralized and equilibrium-point-oriented methods \cite{Stott_Powersystemdynamic_1979,Sastry_Hierarchicalstabilityalert_1980,Chiang_Directstabilityanalysis_1995}, including poor scalability, limited privacy protection, and high computational burden.
	
	However, translating this theory into a practical tool for grid operators requires addressing critical implementation challenges. Part II addresses these challenges by developing implementable methods and demonstrating their applications. Specifically, we focus on two questions left open in Part I: \emph{1) How can the LDD conditions be systematically verified for heterogeneous power system devices?} Although Part I provided the definition, a general numerical method for nonlinear device models is essential. \emph{2) How can the network-wide coupling condition be verified efficiently for large-scale systems?} Therefore, a scalable verification algorithm is necessary.
	
	To address the first challenge, we propose a method that uses a Krasovskii-type storage function to transform the LDD condition into a linear matrix inequality (LMI) feasibility problem. This method applies to a wide range of device models, including synchronous generators and various inverter-based resources. For the second challenge, we devise a distributed algorithm based on the Alternating Direction Method of Multipliers (ADMM), which decomposes the computational burden while preserving device-model privacy.
	
	Beyond these methodological contributions, this paper demonstrates the practical value of the proposed method through three applications relevant to modern grid operation: stability assessment for systems with multiple equilibria, rapid stability screening under fluctuating operating conditions, and a distributed, privacy-preserving stability assessment paradigm for large-scale systems.
%	As Part I focuses on the theoretical basis, Part II is dedicated to methods and practical applications of our theory in complex power grids. Specifically, we tackle two key challenges: verifying the local delta dissipativity property for heterogeneous power devices and verifying the coupling conditions in large systems. 
	\subsection{Literature Review}
	{Verifying device-level properties, such as dissipativity and passivity, is essential for compositional stability analysis. For linear time-invariant systems, the classical passivity can be tested using positive-realness conditions \cite{kottenstette2014relationships} or the Kalman--Yakubovich--Popov lemma \cite{gusev2006kalman}. However, for nonlinear systems, verifying dissipativity is considerably more challenging, with no general method available. Most existing approaches in power systems rely on constructing storage functions based on physical intuition \cite{Fiaz_portHamiltonianapproachpower_2013a,Stegink_unifyingenergybasedapproach_2016}. Few systematic techniques, to the best of our knowledge, include the Hill–Moylan Lemma for classical dissipativity \cite{hill1976stability} and its variation for equilibrium-independent dissipativity \cite{simpson2018equilibrium}, sum-of-squares verification for polynomial models, and Lyapunov-functional/LMI conditions tailored to particular switching, uncertainty, and sampling structures \cite{xia2021dissipativity,chen2020passivity}. However, these methods rely on different model classes and certificate parameterizations. Adapting these frameworks to heterogeneous nonlinear devices in power systems remains a challenge.}
	
	{Verifying interconnection properties, such as the proposed coupling condition, is another critical aspect of compositional stability analysis. Since stability is an inherent property of the entire system, verifying certain interconnection properties is essential to avoid overly conservative results \cite{8890862}. Efficiently verifying interconnection properties in a scalable manner is crucial, particularly in large-scale power grids. Distributed algorithms, such as those based on ADMM \cite{boyd2011distributed}, have been widely applied to power system optimization problems (see \cite{molzahn2017survey} and the references therein), but their deployment in stability analysis is relatively underexplored. Some pioneer works have investigated the inherent connection between stability and optimization and proposed distributed algorithms using dual decomposition \cite{topcu2009compositional} and ADMM \cite{meissen2015compositional} to certify interconnection properties for classical dissipativity. The latter framework also covers nonlinear sum-of-squares formulations and extensions involving equilibrium-independent dissipativity and integral quadratic constraints. ADMM has additionally been applied to distributed stabilizing and $H_\infty$ controller synthesis based on local dissipation inequalities and interconnection conditions \cite{ahmadi2018distributed}. However, it is unclear whether they can be extended to the local delta dissipativity and be equilibrium-free.}
	\subsection{Contributions and Novelty}
%	In this paper, we first propose a method to verify the local delta dissipativity condition based on Krasovskii's type storage function, which applies to a wide variety of nonlinear devices. This ensures that the local stability conditions developed in Part I are applicable to the heterogeneous components of modern power systems. Next, we extend the approach in \cite{meissen2015compositional} and propose a distributed algorithm based on ADMM to verify the coupling condition developed in Part I. This contributes to reducing computational burdens and addressing privacy concerns in large-scale systems.
%	Finally, we demonstrate three key applications of our theory and methods, including stability assessment toward multiple equilibria, assessment under varying operating conditions, and a distributed stability assessment framework for large-scale systems. 
	As an extension of Part I \cite{YP_part-I}, Part II focuses on algorithmic realization of LDD-based stability certification and its practical use in power systems. Relative to Part I and existing dissipativity-based methods, the main contributions are summarized below.
	\begin{itemize}
		\item \textbf{Unified device-level LDD certification for heterogeneous nonlinear power system models.} We establish a unified LDD verification framework for a broad class of power-system devices. By adopting a Krasovskii-type storage function tailored to device-level dynamics, LDD verification is converted into a (robust) LMI feasibility problem. Unlike ad-hoc or model-specific constructions, the procedure supports various industry-relevant device models (e.g., synchronous generators and multiple inverter types) through explicit input-output transformations.
		\item \textbf{Distributed coupling condition verification tailored to equilibrium-free LDD analysis.} We design a scalable distributed algorithm to verify the network-wide coupling condition under an ADMM decomposition framework. The algorithm integrates a \textit{p}-check mechanism, a relaxation variable for infeasibility handling, and an adaptive penalty update rule to improve numerical robustness and convergence speed. Compared with distributed certification methods for classical dissipativity, this extends compositional verification to the equilibrium-free LDD setting developed in Part I.
		\item \textbf{Application-level demonstration of practical advantages.} We demonstrate how the proposed framework supports three operationally relevant scenarios: stability assessment with multiple equilibria, rapid screening under varying operating conditions, and privacy-preserving distributed assessment for large-scale systems. Compared with equilibrium-point-based transient or small-signal analyses, the proposed workflow certifies stability over equilibrium sets within predefined LDD regions and enables repeated assessment under operating-point variations without recomputing global linearization. Case studies on modified IEEE 9-bus, 39-bus, and 118-bus benchmarks validate effectiveness, scalability, and reusability of certified matrices across operating points.
    \end{itemize}

	These results clarify both the methodological novelty and the practical value of the {compositional}, equilibrium-free framework.
	The methodologies and applications presented in this paper provide a practical and scalable toolkit for stability analysis of future power systems.
	
	The rest of Part II is structured as follows. Section II introduces the method for verifying the LDD condition for heterogeneous devices. Section III presents the ADMM-based algorithm for verifying the coupling condition. Section IV explores three key applications of the proposed framework and validates them through case studies on benchmark systems. Section V concludes the paper.

	\section{Verification of Local Delta Dissipativity}
	This section presents a systematic method for verifying the local delta dissipativity of various components in power systems. The approach involves two primary steps: (i) transforming component models into a standardized input-output form, and (ii) validating LDD using a Krasovskii-type storage function. The latter reduces to an LMI problem, enabling straightforward numerical or analytical verification.
    
	\subsection{Model Transformation of Typical Devices in Power Systems}\label{sec:model}
	Many conventional models of device dynamics do not inherently express inputs and outputs in the common direct-quadrature (DQ) voltage-current reference frame. To address this, we introduce a model transformation to reconcile existing device representations with the proposed framework. The third-order SG example was studied in Part I. Below, we demonstrate this process using widely adopted dynamic models for inverter-based resources. These device models are also considered in later system-level case studies.
	\subsubsection{Typical Dynamic Components in Power Systems}\label{sec:dynamic example}
	We first present several widely used dynamic component models for power system stability analysis, including both grid-following and grid-forming inverters.

	\begin{component}[\textit{Grid-following inverters}]
    \label{eg:c5-pll}
		Consider a phase-locked-loop (PLL)-based synchronized inverter, ignoring the fast dynamics of the current and voltage loops. The PLL dynamics, together with the droop power control loop, read \cite{huang2022impacts}
		\begin{equation}\label{eq:pll1}
			\left\lbrace
			\begin{aligned}
				\dot{\xi}&=V_i\sin(\theta_i-\delta_{pll})\\
				\dot{\delta}_{pll}&=\omega_{pll}=K_PV_i\sin(\theta_i-\delta_{pll})+K_I\xi\\
				\tau_1\dot{P}_i&=-P_i+P^{ref}-d_1\omega_{pll}\\
				\tau_2\dot{Q}_i&=-Q_i+Q^{ref}-d_2(V_i-V^{ref}),
			\end{aligned}\right.
		\end{equation}
		where $V_i$ and $\theta_i$ are the magnitude and phase angle of the bus voltage in the DQ coordinate system; $\delta_{pll}$ and $\omega_{pll}$ are the phase angle and frequency obtained from the PLL tracking; $V_i\sin(\theta_i-\delta_{pll})$ is the $d$-axis component of the bus voltage in the local dq coordinate system; $K_P$ and $K_I$ are the gains of the PI control of the PLL; $P_i$ and $Q_i$ are the active and reactive power outputs; $\tau_1$ and $\tau_2$ are the power loop time constants; $d_1$ and $d_2$ are the droop coefficients ($d_1=d_2=0$ for non-droop dynamics); $P^{ref}$, $Q^{ref}$, and $V^{ref}$ are the reference active power, reactive power, and bus voltage, respectively.
		
		Taking $u_i=\col(V_{Di},V_{Qi})$ as the input, we have
		\begin{equation}\label{eq:pll2}
			V_i=\sqrt{V_{Di}^2+V_{Qi}^2},\qquad \theta_i=\arctan(V_{Qi}/V_{Di}).
		\end{equation}
		The output current and power are linked by
		\begin{equation}\label{eq:c5-vsg3}
			\left\lbrace
			\begin{aligned}
				P_i&=I_{Di}V_{Di}+I_{Qi}V_{Qi}\\
				Q_i&=I_{Di}V_{Qi}-I_{Qi}V_{Di}.
			\end{aligned}\right.
		\end{equation}
		This leads to
		\begin{equation}\label{eq:pll3}
			\left\lbrace
			\begin{aligned}
				I_{Di}&=\frac{P_iV_{Di}+Q_iV_{Qi}}{V_{Di}^2+V_{Qi}^2}\\
				I_{Qi}&=\frac{P_iV_{Qi}-Q_iV_{Di}}{V_{Di}^2+V_{Qi}^2}.
			\end{aligned}\right.
		\end{equation}
		
		{
		To make the mapping to the canonical form explicit, define
		\begin{equation}\label{eq:pll-canonical-variables}
			\left\lbrace
			\begin{aligned}
			x_i&=\col(\xi_i,\delta_{pll,i},P_i,Q_i)\\
			u_i&=\col(V_{Di},V_{Qi})\\
			y_i&=\col(-I_{Di},-I_{Qi}).
			\end{aligned}\right.
		\end{equation}
		Using \eqref{eq:pll2}, let $V_i(u_i)=\sqrt{V_{Di}^2+V_{Qi}^2}$ and $\theta_i(u_i)=\arctan(V_{Qi}/V_{Di})$, and define
		\begin{equation*}
			\omega_{pll,i}(x_i,u_i)=K_PV_i(u_i)\sin\!\big(\theta_i(u_i)-\delta_{pll,i}\big)+K_I\xi_i.
		\end{equation*}
		Then the device-specific state vector field is
		{\small
		\begin{equation}\label{eq:pll-canonical-f}
		f_i(x_i,u_i)=
	\begin{bmatrix}
		V_i(u_i)\sin\!\big(\theta_i(u_i)-\delta_{pll,i}\big)\\
		\omega_{pll,i}(x_i,u_i)\\
		\big(-P_i+P^{ref}-d_1\omega_{pll,i}(x_i,u_i)\big)/\tau_1\\
		\big(-Q_i+Q^{ref}-d_2(V_i(u_i)-V^{ref})\big)/\tau_2
	\end{bmatrix}.
		\end{equation}}
		Equation \eqref{eq:pll3} gives the output map
		\begin{equation}\label{eq:pll-canonical-h}
		h_i(x_i,u_i)=-\frac{1}{V_{Di}^2+V_{Qi}^2}
		\begin{bmatrix}
		P_iV_{Di}+Q_iV_{Qi}\\
		P_iV_{Qi}-Q_iV_{Di}
		\end{bmatrix}.
		\end{equation}
		Consequently, the PLL-synchronized inverter is exactly in the form $\dot x_i=f_i(x_i,u_i)$, $y_i=h_i(x_i,u_i)$; its device-specific Jacobians are obtained directly from \eqref{eq:pll-canonical-f} and \eqref{eq:pll-canonical-h}.}
	\end{component}
	
	\begin{component}
	[\textit{Grid-forming inverters with VSG control}]    
    \label{eg:c5-vsg}
		Consider a grid-forming inverter with virtual synchronous generator (VSG) control, which is commonly modeled as \cite{bevrani2014virtual}:
		\begin{equation}\label{eq:c5-vsg1}
			\left\lbrace
			\begin{aligned}
				\dot{\theta}_i&=\omega_i\\
				M_i\dot{\omega}_i&=-D_i\omega_i-P_i+P^{ref}-K_I\theta_i\\
				T_i\dot{V}_i&=-(V_i-V^{ref})+K_Q(Q^{ref}-Q_i)/V_i,
			\end{aligned}\right.
		\end{equation}
		where $V_i$ and $\theta_i$ are the magnitude and phase angle of the bus voltage in the DQ coordinate system; $P_i$ and $Q_i$ are the active and reactive power outputs; $P^{ref}$ and $Q^{ref}$ are the reference values of the output powers; $M_i$ and $D_i$ are the virtual inertia and the frequency damping; $K_I\geq0$ is the frequency integral control coefficient; $T_i$ is the voltage control time constant; $K_Q$ is the reactive power control coefficient.
		
		Taking $y_i=\col(V_{Di},V_{Qi})$ as the output, we have
		\begin{equation}\label{eq:c5-vsg2}
			V_{Di}=V_i\cos\theta_i,\qquad V_{Qi}=V_i\sin\theta_i
		\end{equation}
		
		{Thus, eliminating $P_i$ and $Q_i$ in \eqref{eq:c5-vsg1} using \eqref{eq:c5-vsg3}, we obtain a dynamic subsystem with input $u_i=\col(-I_{Di},-I_{Qi})$, state $x_i=\col(\theta_i,\omega_i,V_i)$, and output $y_i=\col(V_{Di},V_{Qi})$.}
		\end{component}
	
	\begin{component}
    	[\textit{Grid-forming inverters with droop control}] 
    \label{eg:c5-cd}
		The grid-forming inverter with conventional droop control (CD) is modeled as \cite{Zhang_OnlineDynamicSecurity_2015}:
		\begin{equation}\label{eq:c5-cd1}
			\left\lbrace
			\begin{aligned}
				\tau_1\dot{\theta}_i&=-(\theta_i-\theta^{ref})-d_1(P_i-P^{ref})\\
				\tau_2\dot{V}_i&=-(V_i-V^{ref})-d_2(Q_i-Q^{ref}),
			\end{aligned}\right.
		\end{equation}
		where $V_i$ and $\theta_i$ are the magnitude and phase angle of the bus voltage in the DQ coordinate system; $P_i$ and $Q_i$ are the active and reactive power outputs; $P^{ref}$, $Q^{ref}$, $\theta^{ref}$, and $V^{ref}$ are the corresponding reference values; $\tau_1$ and $\tau_2$ are the control time constants; and $d_1$ and $d_2$ are the droop coefficients.
		
		Similarly, with \eqref{eq:c5-cd1}, \eqref{eq:c5-vsg2}, and \eqref{eq:c5-vsg3}, we obtain a dynamic subsystem with input $u_i=\col(-I_{Di},-I_{Qi})$, state $x_i=\col(\theta_i,V_i)$, and output $y_i=\col(V_{Di},V_{Qi})$.
	\end{component}
	
	\begin{component}
    [\textit{Grid-forming inverters with quadratic droop control}]
    \label{eg:c5-qd}
		Another typical inverter-based dynamic model is known as quadratic droop control (QD) in the literature \cite{Simpson-Porco_VoltageStabilizationMicrogrids_2017}. Its dynamics read:
		\begin{equation}\label{eq:egQD}
			\left\lbrace 
			\begin{aligned}
				\tau_{1}\dot{\theta}_i&=-(\theta_i-\theta^{ref})-d_{i1}(P_i-P^{ref})\\
				\tau_{2}\dot{V}_i&=-d_{2}Q_i-V_i(V_i-u^{ref}),
			\end{aligned}\right. 
		\end{equation}
		where $u^{ref}$ is a constant to regulate the steady state that satisfies
		\begin{equation}\label{eq:QDu}
			0=-d_{2}Q^{ref}-V^{ref}(V^{ref}-u^{ref}).
		\end{equation}
		Its voltage dynamics are characterized by a droop relationship between reactive power and the squared voltage, hence the name. The meanings of its parameters and the transformation into DQ coordinates are identical to those of the conventional droop model in the previous example.
	\end{component}
	
	\subsubsection{Typical Static Components in Power Systems}
	This subsection introduces static component models for power system stability analysis. In structure-preserving power system frameworks, many buses (e.g., intermediate network nodes and ZIP loads) are static, meaning their behavior can be represented solely by algebraic relationships between bus voltages and injected currents.
	%	Note that our modeling approach provides a wider application scope than previous works based on the network-reduced model that cannot capture static input-output relations of power devices.
	Several common static components in power systems are examined as follows.
	\begin{component}[\textit{Intermediate Nodes}]\label{eg:c5-lianjie}
		An intermediate node in the network is not connected to any device. Its injected current is always zero. Therefore, it can be considered as a static subsystem with ${u_i}=\col({V}_{Di},{{V}_{Qi}})$ as input, ${{y}_{i}}=-\col{{({{I}_{Di}},{{I}_{Qi}})}}$ as output, and the output function $h_i\equiv\col(0,0)$.
	\end{component}
	
	\begin{component}[\textit{Constant Voltage Sources}]\label{eg:c5-constV}
		In the stability analysis of microgrids and distribution networks, the bus connected to the transmission grid is generally modeled as a constant voltage source, i.e., ${V}_{Di}\equiv {V}_{Di}^0,{V}_{Qi}\equiv {V}_{Qi}^0$. Such nodes can be considered as a static subsystem with input ${{u}_{i}}=-\col{{({{I}_{Di}},{{I}_{Qi}})}}$, output ${{y}_{i}}=\col{{({{V}_{Di}},{{V}_{Qi}})}}$, and the output function $h_i\equiv\col({V}_{Di}^0,{V}_{Qi}^0)$.
	\end{component}
	
	\begin{component}[\textit{ZIP Loads}]\label{eg:c5-zip}
		The ZIP load model is commonly used in power system stability analysis \cite{Kundur_Powersystemstability_1994}. Constant impedance, constant current, and constant power load models can all be considered as special cases of the ZIP load model. Let the bus voltage amplitude be $V_i$, and the active and reactive power loads be $P_i^L$ and $Q_i^L$, respectively. Then the basic expression of the ZIP load model is
		\begin{equation}\label{eq:c5-zip1}
			\left\lbrace\begin{aligned}
				P_i^L &= Z_p^lV_i^2+I_p^lV_i+P_0^l\\
				Q_i^L &= Z_q^lV_i^2+I_q^lV_i+Q_0^l,
			\end{aligned}\right.
		\end{equation}
		where $Z_p^l$, $Z_q^l$, $I_p^l$, $I_q^l$, $P_0^l$, $Q_0^l$ are non-negative constants, which denote the three components of the ZIP load: the constant impedance part, $Z_p^l, Z_q^l$; the constant current part, $I_p^l, I_q^l$; and the constant power part, $P_0^l, Q_0^l$, respectively. 		
		Note that
		\begin{equation*}
			\left\lbrace\begin{aligned}
				-P_i^L &= I_{Di}V_{Di}+I_{Qi}V_{Qi}\\
				-Q_i^L &= I_{Di}V_{Qi}-I_{Qi}V_{Di},
			\end{aligned}\right.
		\end{equation*}
		which yields
		\begin{equation}\label{eq:c5-zip2}
			\left\lbrace
			\begin{aligned}
				I_{Di}&=\frac{-P_i^LV_{Di}-Q_i^LV_{Qi}}{V_{Di}^2+V_{Qi}^2}\\
				I_{Qi}&=\frac{-P_i^LV_{Qi}+Q_i^LV_{Di}}{V_{Di}^2+V_{Qi}^2}\\
				V_i&=\sqrt{V_{Qi}^2+V_{Di}^2}.
			\end{aligned}\right.
		\end{equation}
		Thus, the ZIP load node can be considered as a static subsystem with input ${u}_{i}=\col({V}_{Di},{V}_{Qi})$, output ${y}_{i}=-\col({I}_{Di},{I}_{Qi})$, and output map ${y}_{i}=h_i({u}_{i})$ determined by \eqref{eq:c5-zip1} and \eqref{eq:c5-zip2}.
	\end{component}
	
	\subsection{Verifying LDD of Individual Devices via LMIs}
	After model transformation, we obtain a subsystem model with the required input and output variables. We now consider a general nonlinear dynamic subsystem, where the subscript $i$ is omitted for simplicity.
	\begin{equation}\label{eq:dbus-noi}
		\left\lbrace 	
		\begin{aligned}
			\dot{x}&=f(x,u)\\
			y&=h(x,u).
		\end{aligned}\right. 
	\end{equation}
	The essence of verifying LDD is to find a storage function that locally satisfies the delta dissipativity inequality. To account for the diversity of devices in power systems and enhance the generality of the proposed method, we adopt a Krasovskii-type storage function, i.e., $S(x,u)=f(x,u)^\T\mathcal{P}f(x,u)$ for some positive definite matrix $\mathcal{P}$. The following theorem provides a sufficient condition to verify LDD using such a storage function.
	\begin{theorem}[\textit{LDD of dynamic devices}]\label{pro:c5-lmi}
		{Consider system \eqref{eq:dbus-noi}. If there exist a positive-definite matrix $\mathcal{P}=\mathcal{P}^\text{T}\succ0$ and a scalar $\epsilon>0$ such that, for every $(x,u)\in\mathcal{D}$,}
		{\small
		\begin{equation}\label{eq:c5-lmi}
			\begin{aligned}
				\begin{bmatrix}
					\mathcal{P}\frac{\partial f}{\partial x}+\frac{\partial f^\text{T}}{\partial x}\mathcal{P}+\epsilon I & \mathcal{P}\frac{\partial f}{\partial u}\\
					\frac{\partial f^\text{T}}{\partial u}\mathcal{P}&0
				\end{bmatrix}-
				\begin{bmatrix}
					0&I\\
					\frac{\partial h}{\partial x}&\frac{\partial h}{\partial u}
				\end{bmatrix}^\text{T}
				X
				\begin{bmatrix}
					0&I\\
					\frac{\partial h}{\partial x}&\frac{\partial h}{\partial u}
				\end{bmatrix}\preceq0,
			\end{aligned}
		\end{equation}}
		{then the system satisfies $\operatorname{LDD}(X,\mathcal{D})$.}
	\end{theorem}
	\begin{proof}
		Consider the storage function
		\begin{equation*}
			S(x,u)=f(x,u)^\T \mathcal{P} f(x,u).
		\end{equation*}
		It follows from $\mathcal{P}\succ0$ that
		\begin{equation*}
			\lambda_{\min}(\mathcal{P})\|f(x,u)\|^2\leq S(x,u)\leq\lambda_{\max}(\mathcal{P})\|f(x,u)\|^2,
		\end{equation*}
		where $\lambda_{\min}(\mathcal{P})>0$ and $\lambda_{\max}(\mathcal{P})>0$ represent the minimum and maximum eigenvalues of $\mathcal{P}$, respectively. Hence, $S(x,u)$ satisfies condition 1) of \textbf{Definition 1} in Part I \cite{YP_part-I}.
		
		Calculating the time derivative of $S(x,u)$ yields
		\begin{equation*}
			\begin{aligned}
				\dot{S}(x,u) &= \frac{\partial S(x,u)}{\partial x}\dot{x}+\frac{\partial S(x,u)}{\partial u}\dot{u}\\
				&=
				\begin{bmatrix}
					\dot{x}\\ \dot{u}
				\end{bmatrix}^\T
				\begin{bmatrix}
					\mathcal{P}\frac{\partial f}{\partial x}+\frac{\partial f^\text{T}}{\partial x}\mathcal{P}& \mathcal{P}\frac{\partial f}{\partial u}\\
					\frac{\partial f^\text{T}}{\partial u}\mathcal{P}&0
				\end{bmatrix}
				\begin{bmatrix}
					\dot{x}\\ \dot{u}
				\end{bmatrix}.
			\end{aligned}
		\end{equation*}
		From \eqref{eq:c5-lmi}, it holds that for any $(x,u)\in\mathcal{D}$ 
		{\small
		\begin{equation*}
			\begin{aligned}
				\dot{S}(x,u)\leq
				\begin{bmatrix}
					\dot{x}\\ \dot{u}
				\end{bmatrix}^\T
				\begin{bmatrix}
					0&I\\
					\frac{\partial h}{\partial x}&\frac{\partial h}{\partial u}
				\end{bmatrix}^\text{T}
				X
				\begin{bmatrix}
					0&I\\
					\frac{\partial h}{\partial x}&\frac{\partial h}{\partial u}
				\end{bmatrix}
				\begin{bmatrix}
					\dot{x}\\ \dot{u}
				\end{bmatrix}
				-{\epsilon}\|f(x,u)\|^2.
			\end{aligned}
		\end{equation*}}
%		Notice that
%		\begin{equation*}
%			\begin{bmatrix}
%				0&I\\
%				\frac{\partial h}{\partial x}&\frac{\partial h}{\partial u}
%			\end{bmatrix}
%			\begin{bmatrix}
%				\dot{x}\\ \dot{u}
%			\end{bmatrix}=\begin{bmatrix}
%				\dot{u}\\ \dot{y}
%			\end{bmatrix}
%		\end{equation*}
		Hence, it follows that
		\begin{equation*}
			\dot{S}(x,u)\leq
			\begin{bmatrix}
				\dot{u}\\ \dot{y}
			\end{bmatrix}^\T
			X
			\begin{bmatrix}
				\dot{u}\\ \dot{y}
			\end{bmatrix}
			-{\epsilon}\|f(x,u)\|^2.
		\end{equation*}
		This shows that, for any $(x,u)\in\mathcal{D}$, the storage function $S(x,u)$ satisfies condition 2) in \textbf{Definition 1} (see Part I \cite{YP_part-I}). Hence, the system satisfies LDD($X,\mathcal{D}$).
	\end{proof}
	
	For the linear input-state-output system
	\begin{equation}\label{eq:c5-linearbus}
		\left\lbrace 	
		\begin{aligned}
			\dot{x}&=Ax+Bu\\
			y&=Cx+Du
		\end{aligned}\right. 
	\end{equation}
	condition \eqref{eq:c5-lmi} reduces to a linear matrix inequality independent of $x$ and $u$.
	\begin{equation}\label{eq:c5-lmi-linear}
		\begin{aligned}
			\begin{bmatrix}
				\mathcal{P}A+A^\T \mathcal{P}+\epsilon I & \mathcal{P}B\\
				B^\T \mathcal{P}&0
			\end{bmatrix}-
			\begin{bmatrix}
				0&I\\
				C&D
			\end{bmatrix}^\text{T}
			X
			\begin{bmatrix}
				0&I\\
				C&D
			\end{bmatrix}\preceq0.
		\end{aligned}	
	\end{equation}
	{Hence, if there exist a positive-definite matrix $\mathcal{P}$ and a scalar $\epsilon>0$ such that \eqref{eq:c5-lmi-linear} holds, Theorem \ref{pro:c5-lmi} implies that the linear system \eqref{eq:c5-linearbus} is globally delta dissipative.}

	Now consider a static subsystem
	\begin{equation}\label{eq:sbus-noi}
		y=h(u)
	\end{equation}
	Given a matrix $X$, verifying LDD of \eqref{eq:sbus-noi} can also be transformed into a matrix inequality.
	\begin{theorem}[\textit{LDD of static devices}]
    \label{pro:c5-lmi-static}
		Consider the system \eqref{eq:sbus-noi}, if for any $u\in\mathcal{D}$,
		\begin{equation}\label{eq:c5-Ps}
			\begin{bmatrix}
				I\\\frac{\partial h(u)}{\partial u}
			\end{bmatrix}^\text{T}
			X
			\begin{bmatrix}
				I\\\frac{\partial h(u)}{\partial u}
			\end{bmatrix}\succeq0,\;\;\forall u\in\mathcal{D},
		\end{equation}
		then the system satisfies LDD($X,\mathcal{D}$).
	\end{theorem}
	\begin{proof}
		It directly follows from the definition of local delta dissipativity for static systems.
	\end{proof}

	\begin{remark}
		{The common representation in our models standardizes only the network-facing DQ voltage--current port. It does not require subsystems to have identical state dimensions, dynamic equations, parameters, or certified regions. For each dynamic device, Theorem \ref{pro:c5-lmi} is evaluated using its own $n_i$, $f_i$, $h_i$, Jacobians, storage matrix $\mathcal P_i$, supply-rate matrix $X_i$, and region $\mathcal D_i$, so heterogeneous devices generate different local LMI constraints. Static devices are treated separately by Theorem \ref{pro:c5-lmi-static}. The shared requirements are that the device-specific maps admit the smooth input--state--output representations above on the region of interest and that their DQ ports are compatible with the well-posed network interconnection assumed in Part I.}
	\end{remark}
	
	%	\begin{remark}
		Combining Theorems \ref{pro:c5-lmi} and \ref{pro:c5-lmi-static} yields a numerical approach to verifying LDD for general nonlinear systems. Given a fixed $(x,u)$, condition \eqref{eq:c5-lmi} becomes a standard LMI with $\mathcal{P}$ and $X$ as variables, which can be solved efficiently. Conversely, given fixed $\mathcal{P}$ and $X$, one can compute a region $\mathcal{D}$ over which LDD holds. This can be achieved via sample-based verification with theoretical guarantees \cite{esfahani2014performance}, since the left-hand side of \eqref{eq:c5-lmi} depends continuously on $(x,u)$. One can also estimate $\mathcal{D}$ via perturbation analysis and leverage tools from robust optimization. %It will be our future work to explore better approaches to estimate $\mathcal{D}$.

	To conclude this subsection, we present analytical derivations of LDD for several simple devices. Following \cite{simpson2018equilibrium}, we partition $X\in\rr^{2m\times 2m}$ as
	\begin{equation*}
		X=\begin{bmatrix}
			Q&S\\S^\T&R
		\end{bmatrix},
	\end{equation*}
	where $Q\in\rr^{m\times m}$ and $R\in\rr^{m\times m}$ are symmetric.
	
	\begin{proposition}[\textit{Intermediate nodes and constant voltage sources}]\label{pro:lianjie}
		The intermediate node in \ref{eg:c5-lianjie} and the constant voltage source in \ref{eg:c5-constV} are LDD($X,\rr^2$) for any $X$ with $Q\succeq0$.
	\end{proposition}
	\begin{proof}
		Since $\frac{\partial h}{\partial u}=0$, we have
		\begin{equation*}
			\begin{bmatrix}
				I\\\frac{\partial h}{\partial u}
			\end{bmatrix}^\text{T}
			\begin{bmatrix}
				Q&S\\S^\T&R
			\end{bmatrix}
			\begin{bmatrix}
				I\\\frac{\partial h}{\partial u}
			\end{bmatrix}=
			\begin{bmatrix}
				I\\0
			\end{bmatrix}^\text{T}
			\begin{bmatrix}
				Q&S\\S^\T&R
			\end{bmatrix}
			\begin{bmatrix}
				I\\0
			\end{bmatrix}=Q.
		\end{equation*}
		Invoking Theorem \ref{pro:c5-lmi-static} completes the proof.
	\end{proof}
	
	%	\begin{proposition}[Constant Voltage Source]\label{pro:constV}
		%		The constant voltage source in Example \ref{eg:c5-constV} is LDD($X,\rr^2$) for any $X$ with $Q\succeq0$.
		%	\end{proposition}
	%	\begin{proof}
		%		The proof is identical to Proposition \ref{pro:lianjie}.
		%	\end{proof}
	
	\begin{proposition}[\textit{Constant impedance loads}]\label{pro:constZ}
		The constant impedance load is a special case of the ZIP load in Example \ref{eg:c5-zip}, with only $Z_p^l$ and $Z_q^l$ nonzero. It is LDD($X,\rr^2$) for any $X$ with $Q\succeq\begin{bmatrix}
			-Z_p^l&0\\0&-Z_p^l
		\end{bmatrix}$, $S=\frac{1}{2}I$, and $R\succeq0$.
	\end{proposition}
	\begin{proof}
		Note that
		\begin{equation*}
			\begin{bmatrix}
				I\\\frac{\partial h}{\partial u}
			\end{bmatrix}^\text{T}
			\begin{bmatrix}
				Q&S\\S^\T&R
			\end{bmatrix}
			\begin{bmatrix}
				I\\\frac{\partial h}{\partial u}
			\end{bmatrix}
			=Q+S\frac{\partial h}{\partial u}+\frac{\partial h}{\partial u}^\T S^\T+\frac{\partial h}{\partial u}^\T R\frac{\partial h}{\partial u}.
		\end{equation*}
		For the constant impedance load, we have
		\begin{equation*}
			\frac{\partial h}{\partial u}=\begin{bmatrix}
				Z_p^l&Z_q^l\\-Z_q^l&Z_p^l
			\end{bmatrix}.
		\end{equation*}
		Letting $S=\frac{1}{2}I$ and $R\succeq0$, we have
		\begin{equation*}
			\begin{bmatrix}
				I\\\frac{\partial h}{\partial u}
			\end{bmatrix}^\text{T}
			\begin{bmatrix}
				Q&S\\S^\T&R
			\end{bmatrix}
			\begin{bmatrix}
				I\\\frac{\partial h}{\partial u}
			\end{bmatrix}
			\succeq Q+\begin{bmatrix}
				Z_p^l&0\\0&Z_p^l
			\end{bmatrix}\succeq0,
		\end{equation*}
		which completes the proof.
	\end{proof}
	
	\section{Verification of the Coupling Condition}
	This section presents a systematic method for verifying \textbf{Condition 3} (see Part I \cite{YP_part-I}) for system-wide stability. We first show that this condition can be transformed into an optimization problem.% For large-scale systems, we propose an ADMM-based algorithm to verify the coupling condition in a distributed computational framework.
	\subsection{Optimization-based Verification and Solution Approach}
	The coupling condition can be formulated as an optimization problem for feasibility checking. This optimization problem is then converted into a tractable form. 
    
    Define the matrix-valued function $l^c$ by
	\begin{equation*}
		l^c(X_1,\dots,X_N):=\begin{bmatrix}
			-C\\I
		\end{bmatrix}^\text{T}P_\pi^\text{T}\text{blkdiag}(X_1,\dots,X_{N})P_\pi\begin{bmatrix}
			-C\\I
		\end{bmatrix}
	\end{equation*}
	Here, $X_i$ is the quadratic supply-rate matrix of subsystem $i$, while $C$ and $P_\pi$ are defined in \textbf{Condition 3} of Part I. The coupling condition holds if there exist $X_1,\dots,X_N$ such that $l^c(X_1,\dots,X_N)\preceq0$. In addition, each $X_i$ must satisfy its local LDD constraint, represented by
	\begin{equation*}
		\mathcal{L}_i:=\left\{X_i\mid\text{subsystem $i$ satisfies $\operatorname{LDD}(X_i,\mathcal{D}_i)$}\right\}.
	\end{equation*}
	Combining these requirements, the coupling condition reduces to solving the following feasibility problem:
	\begin{equation}\label{eq:coupling condition}
		\begin{aligned}
			\mathmakebox[3em][c]{\min_{X_i}}\quad & 1\\
			\mathmakebox[3em][c]{\text{s.t.}}\quad & X_i\in\mathcal{L}_i,\quad \forall i\in\mathcal{V},\\
			& l^c(X_1,\dots,X_N)\preceq0.
		\end{aligned}
	\end{equation}
	
	For small to medium-scale systems, \eqref{eq:coupling condition} can be solved centrally. However, centralized computation may become intractable for large-scale systems due to dimensionality issues. To address this, we propose a distributed computational framework based on the ADMM algorithm \cite{boyd2011distributed}, a well-known distributed optimization technique. This method decentralizes the computation of $X_i$ across subsystems while ensuring global convergence, enabling efficient verification of coupling conditions in large networks.

	{We duplicate the local variables $X_i$ to obtain the following equivalent form of \eqref{eq:coupling condition}.}
	\begin{equation}\label{eq:general-admm}
		\begin{aligned}
			\mathmakebox[3em][c]{\min_{X_i, Z_i}}\quad & 1\\
			\mathmakebox[3em][c]{\text{s.t.}}\quad & X_i\in\mathcal{L}_i,\quad \forall i\in\mathcal{V},\\
			& l^c(Z_1,\dots,Z_N)\preceq0,\\
			& X_i=Z_i,\quad \forall i\in\mathcal{V}.
		\end{aligned}
	\end{equation}
	
	The optimization problem \eqref{eq:general-admm} can be solved by ADMM with a convergence guarantee as long as the local constraint set $\mathcal{L}_i$ is convex \cite{boyd2011distributed}. Next, we focus on the special case that considers the Krasovskii-type storage function for LDD. Note, however, that the proposed algorithm can readily adapt to different local sub-problems as long as $\mathcal{L}_i$ is convex.
	
	Now consider Theorem \ref{pro:c5-lmi} for constructing local constraints. Define a matrix-valued function $l_i^d$ for each dynamic bus
	\begin{equation*}
		l_i^d(X_i,\mathcal{P}_i):=
		\begin{bmatrix}
			\mathcal{P}_iA_i+A_i^\T \mathcal{P}_i+{\epsilon} I & *\\
			B_i^\T\mathcal{P}_i&0
		\end{bmatrix}-
		*^\text{T}
		X_i
		\begin{bmatrix}
			0&I\\
			C_i&D_i
		\end{bmatrix},
	\end{equation*}
	and $l_i^s$ for each static bus
	\begin{equation*}
		l_i^s(X_i):=
		-
		\begin{bmatrix}
			I\\
			D_i
		\end{bmatrix}^\text{T}
		X_i
		\begin{bmatrix}
			I\\
			D_i
		\end{bmatrix},
	\end{equation*}
	where $*$ stands for the symmetric part, $A_i$, $B_i$, $C_i$, and $D_i$ are state-input-dependent Jacobian matrices as defined in \eqref{eq:c5-lmi} and \eqref{eq:c5-Ps}.
	Then \eqref{eq:general-admm} becomes
	
	\begin{equation}\label{eq:ADMM-ps}
		\begin{aligned}
			\mathmakebox[3em][c]{\min_{\mathcal{P}_i,X_i,Z_i}}\quad & 1\\
			\mathmakebox[3em][c]{\text{s.t.}}\quad & l^d_i(X_i,\mathcal{P}_i)\preceq0,\quad \mathcal{P}_i\succ0,
			\quad \forall i\in\mathcal{V}_1,\\
			& l^s_i(X_i)\preceq0,\quad \forall i\in\mathcal{V}_2,\\
			& l^c(Z_1,\dots,Z_N)\preceq0,\\
			& X_i=Z_i,\quad \forall i\in\mathcal{V}.
		\end{aligned}
	\end{equation}
	
	{We associate a relaxation variable $t_i$ with each local constraint to introduce an explicit local feasibility margin.}
	\begin{equation}\label{eq:ADMM-ps-relax}
		\begin{aligned}
			\mathmakebox[3em][c]{\min_{\mathcal{P}_i,X_i,Z_i,t_i}}\quad & \sum_{i\in\mathcal{V}}-t_i\\
			\mathmakebox[3em][c]{\text{s.t.}}\quad & l^d_i(X_i,\mathcal{P}_i)\preceq-t_iI,\quad \mathcal{P}_i\succ0,\\
			& t_i\leq\bar{t},\quad \forall i\in\mathcal{V}_1,\\
			& l^s_i(X_i)\preceq-t_iI,\quad t_i\leq\bar{t},
			\quad \forall i\in\mathcal{V}_2,\\
			& l^c(Z_1,\dots,Z_N)\preceq0,\\
			& X_i=Z_i,\quad \forall i\in\mathcal{V}.
		\end{aligned}
	\end{equation}
	
	Here, $\bar{t}>0$ {is} a uniform upper bound on the local feasibility margins $t_i$, preventing excessive values of $t_i$ while improving numerical stability. {The case study below uses} $\bar t=100\epsilon$, where $\epsilon>0$ is the numerical tolerance used for the strict LMI feasibility. Because the $p$-check is activated only when $\min_i t_i>\epsilon$, the present termination mechanism requires $\bar t>\epsilon$. 
	An {feasible} solution of \eqref{eq:ADMM-ps-relax} with all $t_i>0$ guarantees feasibility in the original problem \eqref{eq:ADMM-ps}. Ideally, the solution achieves $t_i^*=\bar{t}$, balancing constraint relaxation across subsystems. {This relaxation enhances algorithm convergence}, especially when the original problem \eqref{eq:ADMM-ps} is infeasible.
	
	%	\begin{remark}
		For nonlinear subsystems, \eqref{eq:ADMM-ps-relax} involves infinitely many local constraints because $l_i^d$ and $l_i^s$ depend on $x_i$ and $u_i$. These inequalities must hold at every point in the required LDD region $\mathcal{D}_i$. {One option is to formulate robust local constraints and apply robust-optimization techniques \cite{ben2002robust}; another is to construct the certificate from sampled points and then characterize the resulting LDD region directly from the LDD inequality. We use the latter procedure in the case studies in Section \ref{sec:case}.}
		%	\end{remark}
	\subsection{Distributed Algorithm Based on ADMM with p-Check}
    The flow chart of our algorithm is shown in Fig. \ref{fig:algorithm}. A distinctive feature of the distributed algorithm is the $p$-check subroutine, {which serves as an early-termination condition for certificate construction: it can stop the computation once a valid certificate is found, before the residual-based ADMM stopping criterion is met.} To present the algorithm clearly, we first outline its main steps with pseudocode and then discuss the role of the $p$-check in stability certification.
    \begin{figure}[h]
		\centering
		\includegraphics[width=0.9\hsize]{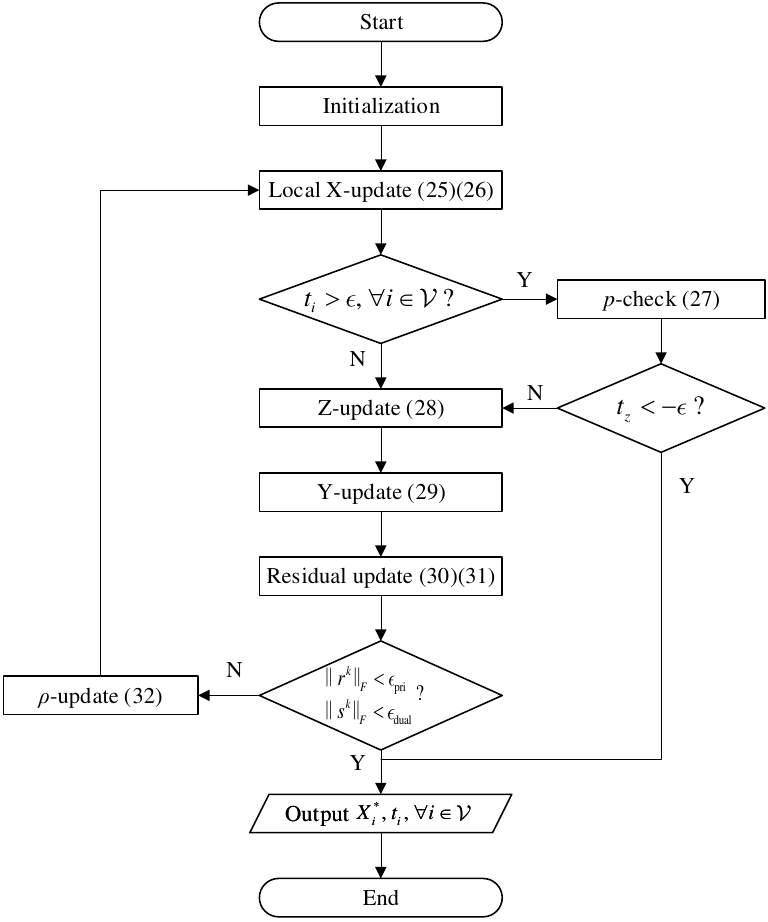}
		\caption{The flow chart of Algorithm \ref{al:1}.}
		\label{fig:algorithm}
	\end{figure}
    
	Define the local Lagrangian for \eqref{eq:ADMM-ps-relax} as
	\begin{equation*}
		L_\rho^i(X_i,Z_i,Y_i):=-t_i+\operatorname{Tr}\!\left(Y_i^\T(X_i-Z_i)\right)
		+\frac{\rho}{2}\|X_i-Z_i\|_{\mathrm F}^2,
	\end{equation*}
	{where $\operatorname{Tr}(\cdot)$ denotes the matrix trace and $\|\cdot\|_{\mathrm F}$ denotes the Frobenius norm. The augmented Lagrangian of \eqref{eq:ADMM-ps-relax} is}
	\begin{equation*}
		L_\rho(X,Z,Y):=\sum_{i\in\mathcal{V}}L_\rho^i(X_i,Z_i,Y_i), 
	\end{equation*}
	where with a little abuse of notation we let $X$, $Z$, and $Y$ denote the collection of $X_i$, $Z_i$, and $Y_i$, respectively. At the $k$-th iteration, our algorithm performs the following six steps based on the augmented Lagrangian.
	
	\subsubsection{X-update}
	Each dynamic bus $i\in\mathcal{V}_1$ solves the following problem
	\begin{equation}\label{eq:X-update-1}
		\begin{aligned}
			\mathmakebox[3em][c]{\min_{\mathcal{P}_i,X_i,t_i}}\quad & L_{\rho^k}^i(X_i,Z_i^k,Y_i^k)\\
			\mathmakebox[3em][c]{\text{s.t.}}\quad & l^d_i(X_i,\mathcal{P}_i)\preceq-t_iI,\\
			& \mathcal{P}_i\succ0,\quad t_i\leq\bar{t}.
		\end{aligned}
	\end{equation}
	Each static bus $i\in\mathcal{V}_2$ solves the following problem
	\begin{equation}\label{eq:X-update-2}
		\begin{aligned}
			\mathmakebox[3em][c]{\min_{X_i,t_i}}\quad & L_{\rho^k}^i(X_i,Z_i^k,Y_i^k)\\
			\mathmakebox[3em][c]{\text{s.t.}}\quad & l^s_i(X_i)\preceq-t_iI,\quad t_i\leq\bar{t}.
		\end{aligned}
	\end{equation}
	The optimal solution yields $X^{k+1}$.
	\subsubsection{p-check}
	{If $t_i>\epsilon$ for every $i\in\mathcal{V}$, we perform the $p$-check by solving}
	\begin{equation}\label{eq:p-check}
		\begin{aligned}
			\mathmakebox[3em][c]{\min_{p_i,t_z}}\quad & t_z\\
			\mathmakebox[3em][c]{\text{s.t.}}\quad & l^c(p_1X_1^{k+1},\dots,p_NX_N^{k+1})\preceq t_zI,\\
			& p_1+\cdots+p_N=N,\quad p_i>0,\quad \forall i\in\mathcal{V}.
		\end{aligned}
	\end{equation}
	The 1-norm constraint bounds the positive weights $p_i$. If the optimum satisfies $t_z<-\epsilon$, the $p$-check passes and the iteration terminates. {Once the local dissipativity conditions are satisfied, a successful $p$-check directly verifies Condition 3 in Part I, thereby providing a valid stability certificate and justifying early termination.}
	
	\subsubsection{Z-update}
	If $p$-check fails or is not performed, we do the Z-update. The system coordinator solves the following problem.
	\begin{equation}\label{eq:Z-update}
		\begin{aligned}
			\mathmakebox[3em][c]{\min_{Z_i}}\quad & L_{\rho^k}(X^{k+1},Z,Y^k)\\
			\mathmakebox[3em][c]{\text{s.t.}}\quad & l^c(Z_1,\dots,Z_N)\preceq0.
		\end{aligned}
	\end{equation}
	The optimal solution yields $Z^{k+1}$.
	
	\subsubsection{Y-update}
	We update the dual variable by
	\begin{equation}\label{eq:Y-update}
		Y_i^{k+1}:=Y_i^k+\rho^k(X_i^{k+1}-Z_i^{k+1}),\quad\forall i\in\mathcal{V}.
	\end{equation}
	
	\subsubsection{Residuals update}
	We calculate the primal residual in the $k$-th iteration by
	\begin{equation}\label{eq:r-update}
		r^{k}=\left[X_1^{k+1},\dots,X_N^{k+1}\right]-\left[Z_1^{k+1},\dots,Z_N^{k+1}\right].
	\end{equation}
	{The dual residual at iteration $k$ is}
	\begin{equation}\label{eq:s-update}
		s^{k}=\rho^k\left(\left[Z_1^{k+1},\dots,Z_N^{k+1}\right]-\left[Z_1^{k},\dots,Z_N^{k}\right]\right).
	\end{equation}
	If $\|r^{k}\|_F<{\epsilon_{\mathrm{pri}}}$ and $\|s^{k}\|_F<{\epsilon_{\mathrm{dual}}}$, then we terminate the iteration. Here, ${\epsilon_{\mathrm{pri}}}>0$ and ${\epsilon_{\mathrm{dual}}}>0$ are the error tolerances for the primal and dual residuals, respectively. Guidelines for selecting these tolerances are provided in \cite{boyd2011distributed}.
	
	\subsubsection{$\rho$-update}
	We use the following adaptive penalty parameter, $\rho^k$, at each iteration to improve convergence. 
	\begin{equation}\label{eq:rho-update}
		\rho^{k+1}=\left\lbrace 
		\begin{aligned}
			\tau^{\text{incr}}\rho^k \;\;\;\; &\text{if} \left\|r^k\right\|_{F}>\mu\left\|s^k\right\|_{F}\\
			\rho^k/\tau^{\text{decr}} \;\;\;\; &\text{if} \left\|s^k\right\|_{F}>\mu\left\|r^k\right\|_{F}\\
			\rho^k \;\;\;\; &\text{otherwise,}
		\end{aligned}\right.
	\end{equation}
	{Here, $\mu>1$ sets the residual-imbalance threshold that triggers an update: a smaller $\mu$ makes the rule more responsive. The factors $\tau^{\mathrm{incr}}>1$ and $\tau^{\mathrm{decr}}>1$ determine the multiplicative increase and decrease of $\rho$, respectively.}
	
	The pseudo-code of our algorithm is summarized below in Algorithm~\ref{al:1}.
    
	\begin{algorithm}[htb]
		\caption{ADMM-based distributed verification of the coupling condition}
		\label{al:1}
		\KwIn{{LMI numerical tolerance $\epsilon>0$; residual tolerances $\epsilon_{\mathrm{pri}}>0$ and $\epsilon_{\mathrm{dual}}>0$; adaptive-penalty parameters $\mu>1$, $\tau^{\mathrm{incr}}>1$, and $\tau^{\mathrm{decr}}>1$; local-margin bound $\bar{t}>\epsilon$; maximum iteration count $M$; and initial values $X_i^1$, $Z_i^1$, $Y_i^1$, and $\rho^1>0$.}}
		\KwOut{$X_i^*$ that verifies the local and coupling conditions.}
		\BlankLine
		Initialize $X_i^1$, $Z_i^1$, $Y_i^1$, and $\rho^1$. Set $k=1$;
		
		\For{$k=1$ \textnormal{to} $M$}
		{
			\ForEach{$i\in\mathcal{V}$}{
				$X$-\textit{update}: Solve \eqref{eq:X-update-1} or \eqref{eq:X-update-2} and obtain $X_i^{k+1}$ and $t_i$;
			} 
			\If{$t_i>{\epsilon}$, $\forall i\in\mathcal{V}$}{
				$p$-\textit{check}: Solve \eqref{eq:p-check} and obtain $t_z$;
				
				\If{$t_z<-{\epsilon}$}{
					\Return{$X_i^*=X_i^{k+1}$ and $t_i$}
				}
			}
			$Z$-\textit{update}: Solve \eqref{eq:Z-update} and obtain $Z^{k+1}$;
			
			$Y$-\textit{update}: Update the dual variable by \eqref{eq:Y-update};
			
			\textit{Residuals update}: Update the primal and the dual residuals by \eqref{eq:r-update} and \eqref{eq:s-update}, respectively;
			
			\If{$\|r^{k}\|_F<{\epsilon_{\mathrm{pri}}}$ and $\|s^{k}\|_F<{\epsilon_{\mathrm{dual}}}$}{
				\Return{$X_i^*=X_i^{k+1}$ and $t_i$}
			}
			
			$\rho$-\textit{update}: Update the penalty by \eqref{eq:rho-update} and obtain $\rho^{k+1}$.
		}
	\end{algorithm}

    \subsection{Discussion on the $p$-check subroutine}
    The key idea of the $p$-check subroutine is to opportunistically test whether the current local certificates $\{X_i^{k+1}\}$ already satisfy the coupling condition after each X-update, instead of waiting until the full ADMM convergence. 
    
    Specifically, the $p$-check solves a lightweight optimization problem \eqref{eq:p-check} to determine whether positive weights $p_i$ make the current local matrices satisfy $l^c(\{p_iX_i^{k+1}\})\preceq t_zI$ with $t_z<-\epsilon$. {A successful check establishes the coupling condition and therefore serves as an early-termination condition for the ADMM-based certificate-construction procedure. Once the certificate is available, the remaining iterations toward the residual-based stopping criterion are unnecessary and can be avoided. If the check is unsuccessful, the ADMM updates proceed unchanged. Thus, the $p$-check can prevent prolonged iterations that are no longer needed for stability certification and reduce the time required to obtain a valid certificate. In the 118-bus case in Section IV-C, the $p$-check terminates the certificate construction at iteration 120, whereas the standard residual criterion is first satisfied at iteration 199, corresponding to a 39.70\% reduction in the required iteration count.}

    {To avoid excessive computation in early iterations, when the local solutions may still be far from feasibility, the $p$-check is triggered only when all local relaxation variables satisfy $t_i>\epsilon$. This condition indicates that every local LDD constraint has a positive numerical margin. To further reduce the checking overhead, the $p$-check can be performed periodically, for example once every five or ten ADMM iterations, rather than at every eligible iteration. This retains regular opportunities for timely early termination while avoiding a $p$-check at every iteration.}

	{
	\subsection{Discussion on Computational Complexity}\label{sec:complexity}
	Let $N$ be the total number of buses, $N_a$ the number of active subsystems (excluding intermediate nodes since they are globally delta dissipative), $n_i$ the state dimension of dynamic subsystem $i$, and $M_i$ the number of local LDD test points. For reference, a dense primal-dual interior-point step for a semi-definite programming (SDP) with $q$ scalar decision variables and positive-semidefinite blocks of orders $\{r_b\}$ has the arithmetic estimate \cite{vandenberghe1996semidefinite}
	\begin{equation}\label{eq:sdp-complexity}
	\mathcal F(q;\{r_b\})=\mathcal O\!\left(q^3+q^2\sum_b r_b^2+q\sum_b r_b^3\right).
	\end{equation}
	For a dynamic-bus $X$-update, $q_i=n_i(n_i+1)/2+11$, with one block of order $n_i$ and $M_i$ LDD blocks of order $n_i+2$. For a static-bus update, $q_i=11$ and the $M_i$ local blocks have order 2. The $Z$-update has $10N_a$ scalar variables and one coupling block of order $2N_a$, whereas the $p$-check has $N+1$ scalar variables and the same coupling-block order. Thus, if $\mathcal C_{X_i}$, $\mathcal C_Z$, and $\mathcal C_p$ denote the corresponding per-solve costs obtained from \eqref{eq:sdp-complexity}, the serial certificate-construction cost is
	\begin{equation}\label{eq:certificate-complexity}
	\mathcal O\!\left(K_X\sum_{i\in\mathcal V_1\cup\mathcal V_2}\mathcal C_{X_i}+K_Z\mathcal C_Z+K_p\mathcal C_p+K_XN\right),
	\end{equation}
	where $K_X$, $K_Z$, and $K_p$ are the numbers of the respective solves. Because the local updates are independent, their sum can ideally be replaced by the slowest local update with parallel computation; the coordinator cost remains. Under the dense estimate, the coordinator step scales at most as $\mathcal O(N^4)$ per Newton step, although sparse SDP implementations can exploit additional structure.

	For comparison, direct small-signal analysis linearizes the DAE as $E\dot\xi=J\xi$. If $n_x=\sum_{i\in\mathcal V_1}n_i$ and $n_{\mathrm{alg}}$ is the number of algebraic variables, a dense full-spectrum QZ calculation for the generalized pencil of order $n_x+n_{\mathrm{alg}}$ requires $\mathcal O((n_x+n_{\mathrm{alg}})^3)$ operations \cite{moler1973algorithm}, excluding equilibrium solution and Jacobian construction. Sparse or partial-spectrum eigensolvers have different costs. Here, $N$ counts bus-level subsystems, while the coordinator dimension is governed by $N_a\leq N$, the number of active supply-rate blocks; in contrast, $n_x+n_{\mathrm{alg}}$ counts all differential and algebraic variables in the full DAE. Because one bus generally contributes several such variables, $n_x+n_{\mathrm{alg}}$ is typically several times $N$. For example, in the 118-bus case in Section IV-C, $n_x+n_{\mathrm{alg}}=424$, $N=118$, and the coordinator coupling LMI has order $2N_a=216$. 
	
	We remark that the proposed methods construct a regional, compositional certificate, whereas eigenvalue analysis tests one specified equilibrium. The advantages are dimensional decomposition, preservation of detailed local-model privacy, and reuse. Once a valid certificate and network model are fixed and a candidate operating point is given, the repeated membership check costs $\mathcal O(\max_{i}(n_i+2)^3)$ under ideal parallel execution, without recomputing a global DAE spectrum.
	}

	\section{Applications}\label{sec:case}
	%In this section, we demonstrate three typical applications of our stability assessment method. The first application is stability assessment toward multiple equilibria, relevant for scenarios where multiple equilibrium points exist under fixed system parameters. Our method enables simultaneous stability evaluation across all equilibria. The second application is stability assessment under varying operating conditions, suited for scenarios with fluctuating parameters—such as load changes or variable renewable energy output—which cause shifts in equilibrium. Our approach facilitates rapid stability assessment as operating conditions change. The third application is distributed stability assessment of large-scale systems, which applies to power systems with massive interconnected devices where centralized methods are computationally burdensome or raise privacy concerns. Our method supports distributed stability assessment, requiring only the exchange of $X_i$ between buses and the coordinator, ensuring scalability and privacy. 
	%
	%In the following subsections, we detail each application and validate our approach using the IEEE 9-bus, 39-bus, and 118-bus benchmark systems, respectively, demonstrating the adaptability of our method across various power system scales and configurations. We remark that although these applications are demonstrated separately, they can be mixed in real cases. For example, the distributed method can be used to assess the stability of multiple equilibria in fixed operating conditions and/or shifting equilibria in varying operating conditions.
	
	This section demonstrates three practical applications of our stability assessment framework using IEEE benchmark systems. Each case highlights distinct advantages of the method while preserving its core theoretical principles. The first application is a stability assessment of multiple equilibria, relevant for scenarios in which multiple equilibrium points exist under fixed system parameters. Our method enables simultaneous stability evaluation across all equilibria. The second application is stability assessment under varying operating conditions, suited to scenarios with fluctuating parameters---such as load changes or variable renewable energy output---that shift the equilibrium. Our approach facilitates rapid stability assessment as operating conditions change. The third application is distributed stability assessment of large-scale systems, which applies to power systems with massive numbers of interconnected devices, where centralized methods are computationally burdensome or raise privacy concerns. {The distributed procedure does not require individual subsystems to disclose their detailed dynamic models or parameters to the coordinator.}

    In the following subsections, we detail each application and validate our approach using the IEEE 9-bus, 39-bus, and 118-bus benchmark systems, respectively, demonstrating the adaptability of our method across various power system scales and configurations. We note that although these applications are presented separately, they can be combined in real cases. 
	\subsection{Application 1: Stability Assessment for Multiple Equilibria}
	%	Resulting from the nonlinear nature of power system models, multiple equilibria often exist even with fixed parameters. Traditional equilibrium-oriented methods, such as eigenvalue analysis, require examining each equilibrium individually. In contrast, our method can directly certify the stability of the equilibria set, as indicated in Theorem x in the Part I of this two-part paper. This enables simultaneous stability evaluation across all equilibria. Here, we begin with a relatively simple system, i.e., the IEEE 9-bus benchmark, to better illustrate the application and to develop deeper insights into the equilibrium-independent nature of our approach. More complex benchmarks will be explored as we move on to subsequent applications.
	Power systems often exhibit multiple equilibria under fixed parameters due to nonlinearity. Traditional equilibrium-oriented methods, such as eigenvalue analysis, require individual equilibrium evaluations. In contrast, our approach certifies stability for all equilibria within a specified region via \textbf{Theorem 2} (see Part I), enabling simultaneous stability evaluation across multiple equilibria.
	\subsubsection{Case settings}
	Consider the modified IEEE 9-bus power system as shown in Fig. \ref{fig:c5-9bus}. To demonstrate the compatibility of the proposed method with heterogeneous dynamics, two SGs are replaced with different types of inverter-interfaced supplies, i.e., a VSG and a CD. Buses 5, 7, and 9 are constant-impedance loads, while Buses 4, 6, and 8 are intermediate nodes. Subsystems' models can be found in Section \ref{sec:model}, with SG parameters from \cite{anderson2008power} and inverter details in Table \ref{tab:c5-9bus-vsg}.
	\begin{figure}[htb]
		\centering
		\includegraphics[width=1\hsize]{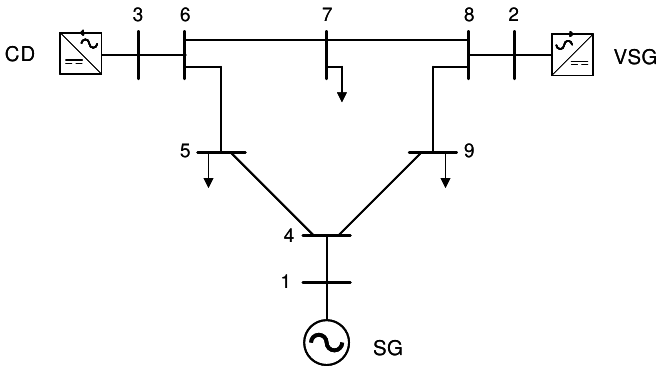}
		\caption{The modified IEEE 9-bus system.}
		\label{fig:c5-9bus}
	\end{figure}
	%	The network and load parameters of the IEEE 9-bus system are obtained from the MATPOWER toolbox\cite{zimmerman2010matpower}. Parameters of the SG at Bus 1 can be found in \cite{anderson2008power}. Parameters of the VSG and the CD are given in Table \ref{tab:c5-9bus-vsg}.
	
	\begin{table}[h]
		\centering
		\footnotesize
		\caption{Parameters of VSG and CD in the IEEE 9-bus system.}
		\label{tab:c5-9bus-vsg}
		\begin{tabular}{ccc}
			\toprule
			Device &Parameters&Values (p.u.)\\
			\midrule
			VSG& $M$, $D$, $K_Q$, $K_I$, $T$,&20, 0.5, 0.01, 0.5, 5\\ 
			%			\midrule
			CD& $\tau_{1}$, $\tau_{2}$, $d_{1}$, $d_{2}$,  &0.3,   6,    0.01,   0.01  \\
			\bottomrule
		\end{tabular}
	\end{table}
	
	%\begin{table}[h]
	%	\centering
	%	\footnotesize
	%	\caption{Parameters of VSG and CD in the IEEE 9-bus system.}
	%\begin{tabular}{lll}
	%	\toprule
	%	Device &\multicolumn{1}{c}{Parameters}&\multicolumn{1}{c}{Values (p.u.)}\\
	%	\midrule
	%	\multirow{2}*{VSG}& $M$, $D$, $K_Q$, $K_I$, $T$,&20, 0.5, 0.01, 0.5, 5,\\ &$P^{ref}$, $Q^{ref}$, $V^{ref}$ &  1.714, 0.145, 1\\ 
	%	\midrule
	%	\multirow{2}*{CD}& $\tau_{1}$, $\tau_{2}$, $d_{1}$, $d_{2}$, $P^{ref}$, &0.3,   6,    0.01,   0.01, 0.85, \\
	%	& $ Q^{ref}$, $\theta^{ref}$, $V^{ref}$ &  -0.0365,  0.0833,  1 \\
	%	\bottomrule
	%\end{tabular}
	%	\label{tab:c5-9bus-vsg}
	%\end{table}
	
	\subsubsection{Verification of stability conditions}
		We first verify the local conditions for the three dynamic subsystems, namely SG, VSG, and CD, whose input-state-output models are presented in Section \ref{sec:model}. We consider the storage function of Krasovskii's type $S_i(x_i,u_i)=f_i(x_i,u_i)^\T \mathcal{P}_if_i(x_i,u_i)$, $i\in\mathcal{V}_1$ and verify LDD numerically according to Theorem \ref{pro:c5-lmi}. 
        
		It is found that the SG at Bus 1 is LDD($X_1,\mathcal{D}_1$) with the following $\mathcal{P}_1$ and $X_1$:
        \begin{equation*}
		% \footnotesize
		\mathcal{P}_1=\begin{bmatrix}
			2.6392 & 0.0672 & 2.0569\\
			0.0672 & 0.2524 & -0.3840\\
			2.0569 & -0.3840 & 39.6408
		\end{bmatrix},
	\end{equation*}
	\begin{equation*}
		% \footnotesize
		X_1=\begin{bmatrix}
			46.3119 & 8.2589 & 4.7649 & -11.7205\\
			8.2589 & 76.6718 & 16.9166 & 2.6670\\
			4.7649 & 16.9166 & 3.6030 & -0.1548\\
			-11.7205 & 2.6670 & -0.1548 & 2.5669
		\end{bmatrix}.
	\end{equation*}
    The corresponding LDD region $\mathcal{D}_1$ is 5-dimensional. To visualize $\mathcal{D}_1$, we project it on the 2-dimensional input space, as shown in Fig. \ref{fig:9bus-D1}.
    \begin{figure}[htb]
		\centering
		\includegraphics[width=0.75\hsize]{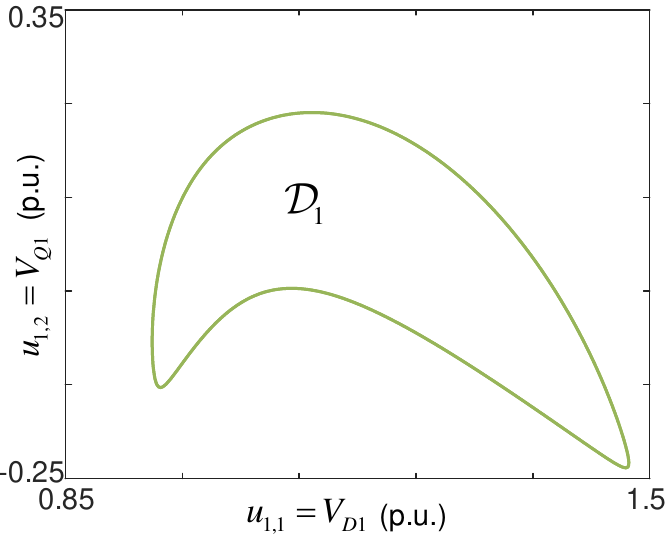}
		\caption{{The cross section of $\mathcal{D}_1$ on the two-dimensional input plane.}}
		\label{fig:9bus-D1}
	\end{figure}
	%To this end, we first solve the power flow under the nominal operation condition of the benchmark\cite{zimmerman2010matpower}, which yields the nominal network input $u_i^0$, $i\in\mathcal{V}$. The nominal state $x_i^0$ is then obtained by solving $f_i(x_i,u_i^0)=0$ for each $i\in\mathcal{V}_1$. We then solve \eqref{eq:c5-lmi} for $\mathcal{P}_i$ and $X_i$ at the nominal input-state equilibrium $(x_i^0,u_i^0)$ to verify delta dissipativity of each dynamic subsystem. The solution results are as follows.
	% Using Krasovskii-type storage functions $S_i(x_i,u_i)=f_i(x_i,u_i)^\T \mathcal{P}_if_i(x_i,u_i)$, we numerically verified delta dissipativity for all dynamic subsystems  through Theorem \ref{pro:c5-lmi} with the following $\mathcal{P}_i$ and $X_i$.
	Similarly, we verified LDD for the VSG at Bus 2 and the CD at Bus 3, with the following $\mathcal{P}_i$ and $X_i$.
	
	\begin{equation*}
		% \footnotesize
		\mathcal{P}_2=\begin{bmatrix}
			0.0387 & 0.1511 & -0.0189\\
			0.1511 & 3.1548 & 0.7851\\
			-0.0189 & 0.7851 & 349.7891
		\end{bmatrix},
	\end{equation*}
	\begin{equation*}
		% \footnotesize
		X_2=\begin{bmatrix}
			1.3050 & 0.2361 & 0.3267 & 6.6024\\
			0.2361 & 1.9280 & -9.8563 & 2.5905\\
			0.3267 & -9.8563 & 36.6184 & -6.4825\\
			6.6024 & 2.5905 & -6.4825 & 14.8924
		\end{bmatrix}.
	\end{equation*}
    Its LDD region $\mathcal{D}_2$ is visualized in Fig. \ref{fig:9bus-D2}.
	\begin{figure}[htb]
		\centering
		\includegraphics[width=0.75\hsize]{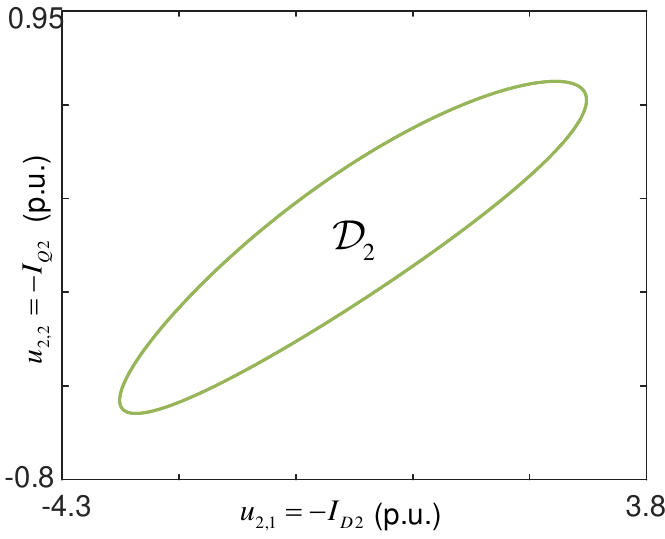}
		\caption{{The cross section of $\mathcal{D}_2$ on the two-dimensional input plane.}}
		\label{fig:9bus-D2}
	\end{figure}
    We also verified LDD for the CD at Bus 3, with the following $\mathcal{P}_i$ and $X_i$.
	\begin{equation*}
		% \footnotesize
		\mathcal{P}_3=\begin{bmatrix}
			133.4954 & 6.7946\\
			6.7946 & 470.9895\\
		\end{bmatrix},
	\end{equation*}
	\begin{equation*}
		% \footnotesize
		X_3=\begin{bmatrix}
			0.0791 & 0.0686 & 1.4438 & -1.7175\\
			0.0686 & 0.0893 & 0.9115 & -5.7121\\
			1.4438 & 0.9115 & -117.3882 & -104.3170\\
			-1.7175 & -5.7121 & -104.3170 & -336.4658
		\end{bmatrix}.
	\end{equation*}
	Its LDD region $\mathcal{D}_3$ is visualized in Fig.  \ref{fig:9bus-D3}. 
    % Static subsystems satisfy global delta dissipativity as proved in Propositions \ref{pro:lianjie} and \ref{pro:constZ}.
	% \begin{figure}[htb]
	% 	\centering
	% 	\includegraphics[width=1\hsize]{fig-9busD}
	% 	\caption{The cross sections of $\mathcal{D}_i$, $i=1,2,3$ on the 2-dimensional input plane. The red dots represent the nominal inputs $u_i^0$ of each dynamic subsystem.}
	% 	\label{fig:9bus-D}
	% \end{figure}
    
    \begin{figure}[htb]
		\centering
		\includegraphics[width=0.75\hsize]{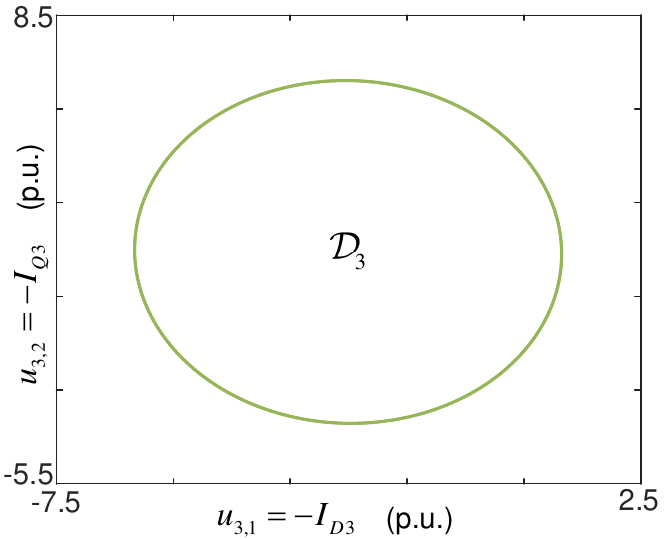}
		\caption{{The cross section of $\mathcal{D}_3$ on the two-dimensional input plane.}}
		\label{fig:9bus-D3}
	\end{figure}
	{The two-dimensional sections in Figs. \ref{fig:9bus-D1}--\ref{fig:9bus-D3} are constructed using a $161\times161$ input grid with the local states fixed at their nominal values. Equilibrium membership is checked independently of this plotting grid by directly substituting the complete local state and input into the local dissipativity condition.}
	
	Now we verify the local condition for the static subsystems, i.e., Bus 4-9. It follows from Proposition \ref{pro:lianjie} and \ref{pro:constZ} that they are globally delta dissipative with a wide choice of $X_i$, $i\in\mathcal{V}_2$. $X_i$ of static subsystems are determined along with $p_i$, $i\in\mathcal{V}_1$, as we solve the LMI of the network coupling condition. Expressions of $X_i$, $i\in\mathcal{V}_2$ are omitted here due to page limit.
	
	\subsubsection{Multiple equilibria assessment}
	We identified five isolated equilibria under nominal parameters, as shown in Table \ref{tab:c5-9bus-multieq}.
	\textbf{Theorem 2} (Part I) ensures that any isolated equilibrium in $\mathcal{D}=\mathcal{D}_1\times\mathcal{D}_2\times\mathcal{D}_3\times\rr^{12}$ is asymptotically stable. For any equilibrium $(x^*,u^*)$ of the system, each subsystem can independently verify whether its local state $(x^*_i,u^*_i)$ lies within $\mathcal{D}_i$ and hence certify the system-wide asymptotic stability.
	
	Equilibrium \#1-\#2 are located within $\mathcal{D}$ while \#3-\#5 do not. Hence, our method directly indicates that equilibrium \#1 and \#2 are asymptotically stable, while no conclusion can be drawn about the other three equilibria.	
	To compare, eigenvalue analysis confirmed stability for equilibrium \#1-\#3 and instability for \#4-\#5. Table \ref{tab:c5-9bus-multieq} summarizes this result. This verifies the correctness of our theory and also demonstrates that our method can be conservative to some extent, as it is only a sufficient, not a necessary, condition for stability.
	
	\begin{table}[h]
		\centering
		\footnotesize
		\caption{Five Equilibria of the IEEE 9-bus system.}
		\label{tab:c5-9bus-multieq}
		\begin{tabular}{cccc}
			\toprule
			E.q. &{$x^*=(\delta_1^*,\omega_1^*,E_{q1}^{\prime *},\theta_2^*,\omega_2^*,V_2^*,\theta_3^*,V_3^*)$}&Location& $\max\Re\lambda$\\
			\midrule
			\#1& {(0.26, 0, 1.02, 0.17, 0, 1.00, 0.08, 1.00)}&$\in\mathcal{D}$& -0.17\\ 
			
			\#2& {(-0.10, 0, 0.98, 5.70, 0, 0.99, 0.06, 0.99)}  &$\in\mathcal{D}$ & -0.17  \\
			\#3& {(0.82, 0, 0.94, -5.18, 0, 0.98, 0.11, 0.99)}  &$\notin\mathcal{D}$ & -0.15\\
			\#4& {(1.95, 0, 0.64, -4.09, 0, 0.95, 0.13, 0.95)}  &$\notin\mathcal{D}$ & 1.62\\
			\#5& {(2.43, 0, 0.06, 0.34, 0, 0.99, 0.093, 0.99)}  &$\notin\mathcal{D}$ & 1.23\\
			\bottomrule
		\end{tabular}
	\end{table}
    
	We also performed a time-domain simulation of the equilibrium step change from \#1 to \#2, which demonstrates the asymptotic stability of these two equilibria. The voltage-magnitude, relative phase-angle, and bus-frequency-deviation trajectories are shown in Figs. \ref{fig:9bus-eq-switch-V}--\ref{fig:9bus-eq-switch-omega}. {The same figure legend is used in all three figures.}
    
	\begin{figure}[htb]
		\centering
		\includegraphics[width=0.95\hsize]{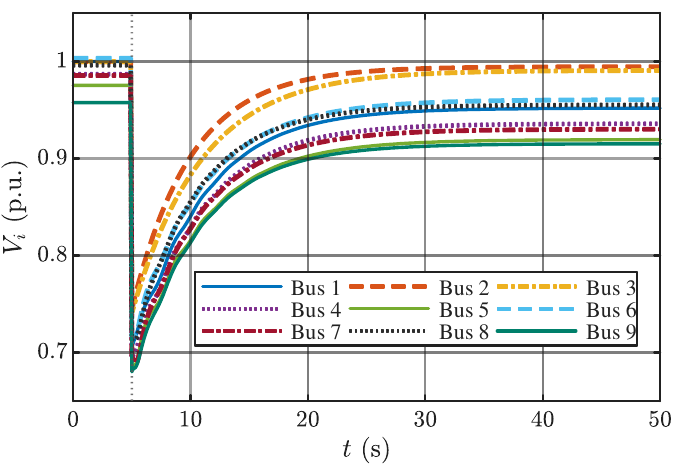}
		\caption{{Voltage-magnitude trajectories of Buses 1--9 following the state disturbance at $t=5$~s.}}
		\label{fig:9bus-eq-switch-V}
	\end{figure}
    
	\begin{figure}[htb]
		\centering
		\includegraphics[width=0.95\hsize]{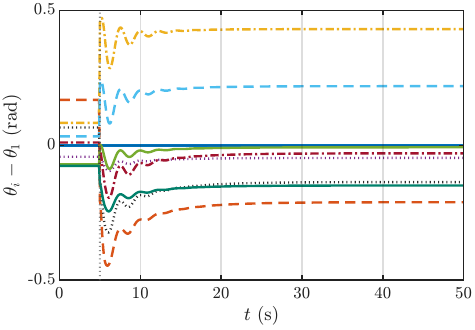}
		\caption{{Relative phase-angle trajectories following the state disturbance at $t=5$~s. All nine buses are shown using the same legend as Fig. \ref{fig:9bus-eq-switch-V}.}}
		\label{fig:9bus-eq-switch-theta}
	\end{figure}
    
    \begin{figure}[htb]
		\centering
		\includegraphics[width=0.95\hsize]{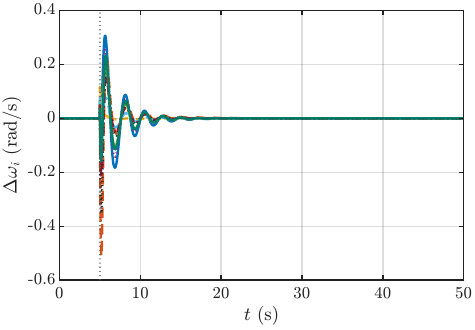}
		\caption{{Frequency-deviation trajectories following the state disturbance at $t=5$~s. All nine buses are shown using the same legend as Fig. \ref{fig:9bus-eq-switch-V}.}}
		\label{fig:9bus-eq-switch-omega}
	\end{figure}
	
	\subsection{Application 2: Stability Assessment under Varying Operating Conditions}
	Power systems are subjected to increasing fluctuation from both loads and renewable energy resources. This leads to variation in the system's parameters and, hence, shifts in equilibrium. Thanks to the equilibrium-free nature of LDD, our method enables rapid stability evaluation under shifting equilibria. Our framework pre-computes the LDD matrix $X_i$ and verifies the equilibrium-independent coupling condition. As parameters evolve, our method only needs to revalidate the local delta dissipativity of the subsystems. 
    % We demonstrate this application by the IEEE 39-bus benchmarks as follows.
	\subsubsection{Revisit the IEEE 9-bus benchmark}
		We scale the loads in the IEEE 9-bus system by a factor $s>0$ to create different operating scenarios. We assume the SG at bus 1 is the balance node, which adjusts its parameters $P_1^m$ and $E_{f1}$ to maintain its terminal voltage and phase at their nominal values, i.e., {$u_1^*=(V_{D1}, V_{Q1})^\T\equiv(1,0)^\T$}. For the VSG and CD, we assume they do not regulate their reference power outputs and that all parameters are fixed. Hence, the network-induced input $u_i$ varies with load changes. 
		With $X_i$ fixed, the network coupling condition holds for all $s$. Hence, only the local delta dissipativity of each dynamic subsystem at the state-input equilibrium $(x_i^*,u_i^*)$ induced by $s$ needs to be checked. {Fig. \ref{fig:9bus-varying} shows that the common compositional certificate guarantees system-wide stability for $s\in(0.79,1.20)$, whereas equilibrium-wise eigenvalue analysis finds pointwise stability for $s\in(0.60,1.80)$. The narrower certified interval reflects the requirement of a common, reusable sufficient certificate over an operating region. Reoptimizing $X_i$ for different operating regions may enlarge the certified interval, at the cost of recomputing the coupling certificate.}
		\begin{figure}[htb]
				\centering
				\includegraphics[width=1\hsize]{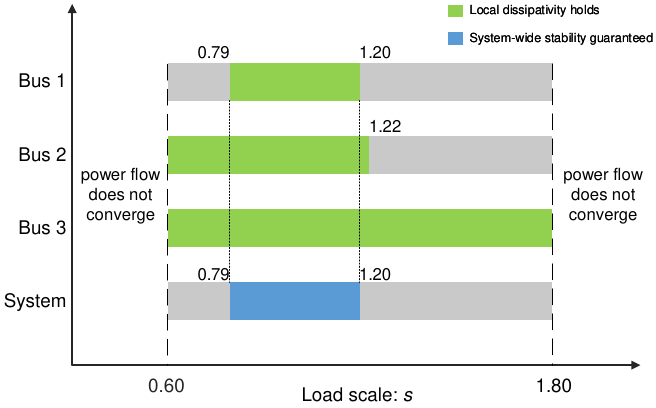}
				\caption{Compositional stability certificate under varying operating conditions. System-wide stability is guaranteed for $s\in(0.79,1.20)$, in which LDD holds for each subsystem. When $s<0.60$, the power flow does not converge due to active power imbalance. When $s>1.80$, the power flow does not converge because the bus voltage falls below 0.9 p.u.}
				\label{fig:9bus-varying}
			\end{figure}
		{
			\begin{remark}[Conservativeness]\label{rem:conservativeness}
				The main sources of conservativeness are: 1) the adopted Krasovskii-type quadratic local storage functions; 2) the block-separable aggregate storage construction and the sufficient network-coupling condition; and 3) the use of regional certificates and, when assessing varying operating conditions, fixed $X_i$ matrices that allow the coupling certificate to be reused. Accordingly, for an equilibrium outside $\mathcal D$, the conclusion is limited to ``not certified''. Less conservative results may be obtained through richer storage-function parameterizations, regional partitioning with multiple certificates, or operating-region-dependent or reoptimized $X_i$ matrices. These refinements trade additional offline computation and potentially reduced certificate reuse for a larger certified region. To place the comparison with equilibrium-wise eigenvalue analysis in context, the two methods provide different and complementary guarantees. Eigenvalue analysis evaluates local small-signal stability at one precomputed equilibrium using its equilibrium-dependent Jacobian. The proposed method, meanwhile, seeks a common compositional certificate that enforces nonlinear local delta-dissipativity inequalities over prescribed regions and certifies every admissible isolated equilibrium contained in $\mathcal D$. The compositional structure keeps certificate construction local and permits independent or parallel local computation without disclosing detailed device models, while the regional common certificate can be reused across compatible equilibria and operating conditions.
			\end{remark}
		}
	\subsubsection{IEEE 39-bus case settings}
		Consider the modified IEEE 39-bus system integrating 13 heterogeneous dynamic devices (SGs and inverters) and 26 static components, as shown in Fig. \ref{fig:c5-39bus}. The models of all dynamic subsystems can be found in Section \ref{sec:model}. {The system contains 16 constant-impedance load buses and 10 passive connection buses.} The inverter parameters are listed in Table \ref{tab:c5-39bus-para}, and the SG parameters can be found in \cite{Pai_EnergyFunctionAnalysis_1989}.
	\begin{figure}[htb]
		\centering
		\includegraphics[width=1\hsize]{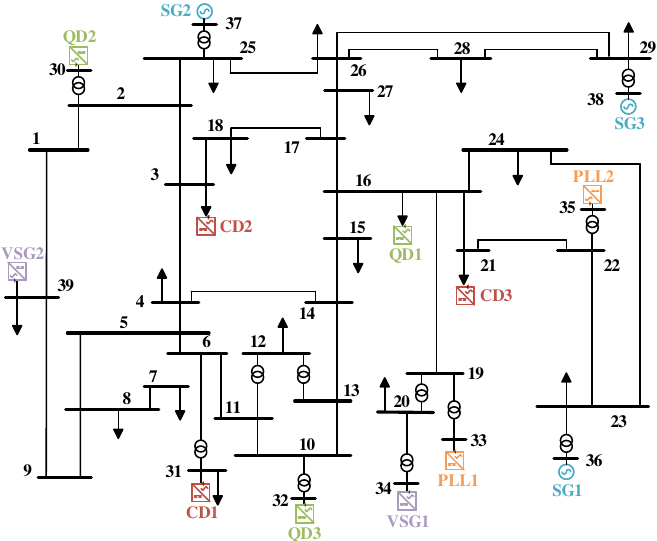}
		\caption{The modified IEEE 39-bus system.}
		\label{fig:c5-39bus}
	\end{figure}
	
	\begin{table}[h]
		\centering
		\footnotesize
		\caption{Parameters of inverters in the IEEE 39-bus system.}
		\label{tab:c5-39bus-para}
		\begin{tabular}{ccc}
			\toprule
			Device & Parameters  &  Values (p.u.)\\ 
			\midrule
			CD1&$\tau_1,\tau_2,d_1,d_2$  & 0.3, 8.1, 0.01, 0.01 \\
			CD2&$\tau_1,\tau_2,d_1,d_2$  & 0.25, 9, 0.01, 0.01 \\
			CD3&$\tau_1,\tau_2,d_1,d_2$  & 0.35, 8.4, 0.01, 0.01 \\
			QD1&$\tau_1,\tau_2,d_1,d_2$  & 0.3, 8, 0.01, 0.01 \\
			QD2&$\tau_1,\tau_2,d_1,d_2$  & 0.2, 7.5, 0.01, 0.01 \\
			QD3&$\tau_1,\tau_2,d_1,d_2$  & 0.4, 6.5, 0.01, 0.01 \\
			PLL1&$K_I,K_P,\tau_1,\tau_2,d_1,d_2$ & 590, 34, 1, 1, 0.5, 0.5\\
			PLL2&$K_I,K_P,\tau_1,\tau_2,d_1,d_2$ & 592, 35, 1.2, 1.2, 0.4, 0.4\\
			VSG1& $M$, $D$, $K_Q$, $K_I$, $T$,&20, 1.5, 0.01, 0.5, 5\\ 
			VSG2& $M$, $D$, $K_Q$, $K_I$, $T$,&50, 2.5, 0.01, 0.8, 7\\ 
			\bottomrule
		\end{tabular}
	\end{table}
	
	\subsubsection{Verifying stability conditions}
	Using Theorem \ref{pro:c5-lmi}, we verify that all dynamic subsystems satisfy LDD with computed $\mathcal{P}_i$, $X_i$, and $\mathcal{D}_i$.
	    Solutions for PLL1 at bus 33 and QD1 at bus 16 are provided below as examples.
		Solutions for other dynamic subsystems are omitted here due to page limit.
		{
			\begin{equation*}
				% \footnotesize
				\mathcal{P}_{PLL1}=\begin{bmatrix}
						827.5740 & -5.8790 & -8.9559 & 2.4322\\
						-5.8790 & 11.0117 & 21.9792 & -10.4929\\
						-8.9559 & 21.9792 & 43.9294 & -20.9534\\
						2.4322 & -10.4929 & -20.9534 & 13.8473
					\end{bmatrix},
			\end{equation*}
	
		\begin{equation*}
				% \footnotesize
				X_{PLL1}=\begin{bmatrix}
						-97.5135 & 58.3279 & 17.5625 & 7.1183\\
						58.3279 & -91.5440 & -11.1970 & -13.0970\\
						17.5625 & -11.1970 & -3.0913 & -1.4267\\
						7.1183 & -13.0970 & -1.4267 & -1.6995	
					\end{bmatrix},
			\end{equation*}
		
		\begin{equation*}
				% \footnotesize
				\mathcal{P}_{QD1}=\begin{bmatrix}
						51.1467 & -3.3244\\
						-3.3244 & 999.9996
					\end{bmatrix},
			\end{equation*}
	
		\begin{equation*}
				% \footnotesize
				X_{QD1}=\begin{bmatrix}
						0.0146 & -0.0001 & -0.0181 & 1.6886\\
						-0.0001 & 0.0148 & -1.3374 & -0.0521\\
						-0.0181 & -1.3374 & -241.0171 & 3.7537\\
						1.6886 & -0.0521 & 3.7537 & -310.3471
					\end{bmatrix}.
			\end{equation*}
		}
	The LDD regions $\mathcal{D}_i$ of PLL1 and QD1 are six- and four-dimensional, respectively. Fig. \ref{fig:39bus-D-1} and \ref{fig:39bus-D-2} show their two-dimensional sections in the input space, along with the inputs ${u_i^*}$ under the nominal operating condition ($s=1$). 
	
	\begin{figure}[htb]
		\centering
		\includegraphics[width=0.8\hsize]{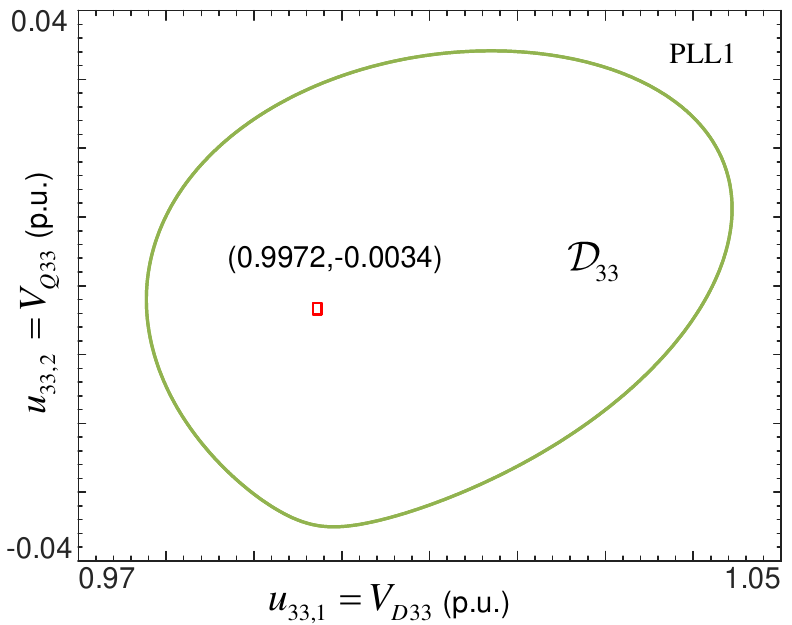}
		\caption{{The LDD region of PLL1 at bus 33 on the input plane. The red dots represent the inputs $u_i^*$ under the nominal operating condition $s=1$.}}
		\label{fig:39bus-D-1}
	\end{figure}

    \begin{figure}[htb]
		\centering
		\includegraphics[width=0.8\hsize]{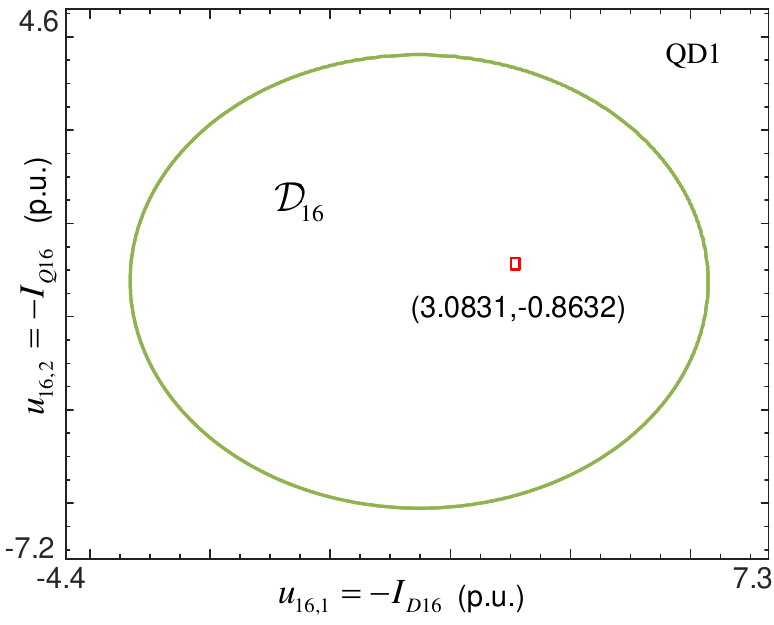}
		\caption{{The LDD region of QD1 at bus 16 on the input plane. The red dots represent the inputs $u_i^*$ under the nominal operating condition $s=1$.}}
		\label{fig:39bus-D-2}
	\end{figure}

	    All static subsystems are globally delta dissipative, and the coupling condition is also verified by solving the feasibility problem \eqref{eq:coupling condition}. Hence, our \textbf{Theorem 2} (Part I) guarantees that any isolated equilibrium in $\mathcal{D}$ is asymptotically stable. 
		{At the nominal operating condition, each subsystem locally verifies that $(x_i^*,u_i^*)$ lies within $\mathcal D_i$, thereby certifying system-wide asymptotic stability. Centralized eigenvalue analysis gives a maximal real part of $-0.05763$, supporting this conclusion.}
	
	\subsubsection{Varying operating conditions}
		Consider different operating conditions as the system load varies, with a scaling factor $s>0$. SG3 at bus 38 is assumed to be the balancing node, which adjusts its $P^m$ and $E_f$ to maintain a constant terminal voltage. Other dynamic subsystems use fixed parameters for all $s$. {We scanned $s\in[0.5,2.0]$ with spacing 0.005. Table \ref{tab:c5-39bus-s} reports the resulting numerical intervals of $s$ over which each dynamic subsystem satisfies the LDD condition. Their intersection is $s\in[0.8535,1.1625]$, over which the fixed certificate establishes system-wide stability. Fig. \ref{fig:39bus-timedomain} shows the time-domain response of the nonlinear system when $s$ switches among 1, 0.86, 1, and 1.16, all within the common certified interval supporting our result.}
	
	{
	\begin{table}[h]
		\centering
		\footnotesize
		\caption{{Numerically identified range of $s$ over which each dynamic subsystem satisfies the LDD condition with fixed $X_i$.}}
		\label{tab:c5-39bus-s}
		\begin{tabular}{cccc}
			\toprule
			Device & Range of $s$ & Device & Range of $s$\\ 
			\midrule
			SG1 & [0.6874, 1.3624] & QD2 & [0.5000, 1.7655] \\
			SG2 & [0.8258, 1.1669] & QD3 & [0.5000, 1.4358] \\
			SG3 & [0.8535, 1.1672] & PLL1 & [0.8252, 1.1627]\\
			CD1 & [0.5000, 2.0000] & PLL2 & [0.8172, 1.1635]\\
			CD2 & [0.5000, 1.1654] & VSG1 &[0.6269, 1.2047]\\
			CD3 & [0.5000, 1.4824] & VSG2 & [0.7696, 1.1914]\\
			QD1 & [0.7392, 1.5297] & & \\
			\bottomrule
		\end{tabular}
	\end{table}
	}

	{The common-certificate interval $s\in[0.8535,1.1625]$ corresponds to load variations of approximately $-14.6\%$ to $+16.3\%$ around the nominal operating condition.}
	
	\begin{figure}[htb]
		\centering
		\includegraphics[width=1\hsize]{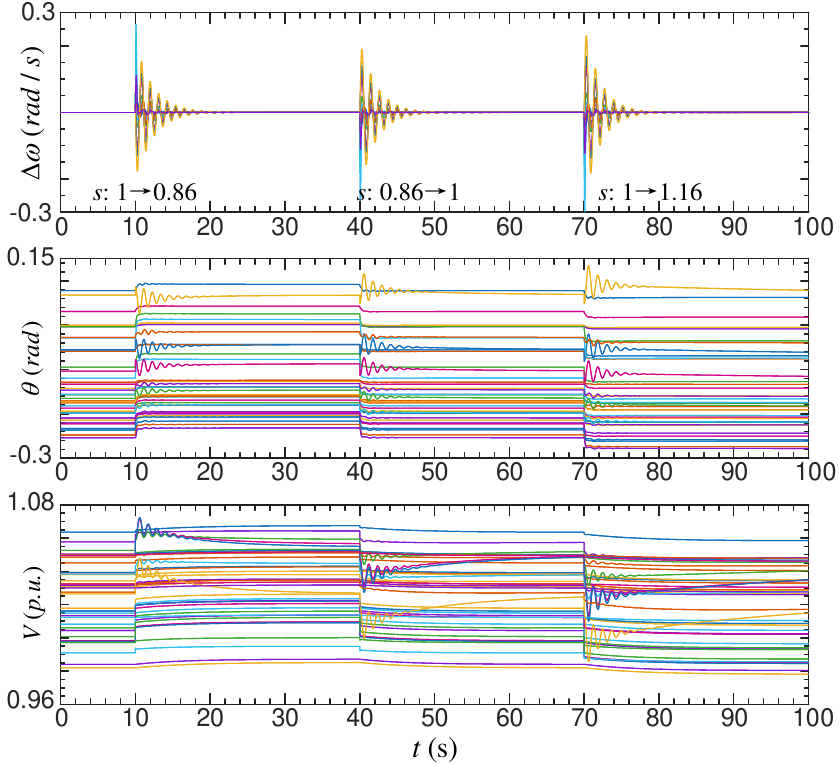}
		\caption{{Nonlinear time-domain response as the system load factor $s$ switches from 1 to 0.86 at $t=10$ s, from 0.86 to 1 at $t=40$ s, and from 1 to 1.16 at $t=70$ s.}}
		\label{fig:39bus-timedomain}
	\end{figure}

	{
	\subsubsection{Three-phase-fault validation}
	To complement the equilibrium-shift tests, we simulate the nonlinear power system model under a balanced three-phase fault at Bus 17. The fault is applied at $t=1.0$~s with impedance $0.001+\mathrm{j}0.01$ p.u., cleared after 100~ms, and the prefault topology is restored at $t=1.1$~s. Fig. \ref{fig:39bus-fault} shows the voltage and frequency of each bus during the event. The minimum bus-voltage magnitude is 0.65385 p.u. and the maximum bus-frequency deviation is 0.05613~Hz. The trajectories recover after clearing, which supports our stability certification.

	\begin{figure}[htb]
		\centering
		\includegraphics[width=\linewidth]{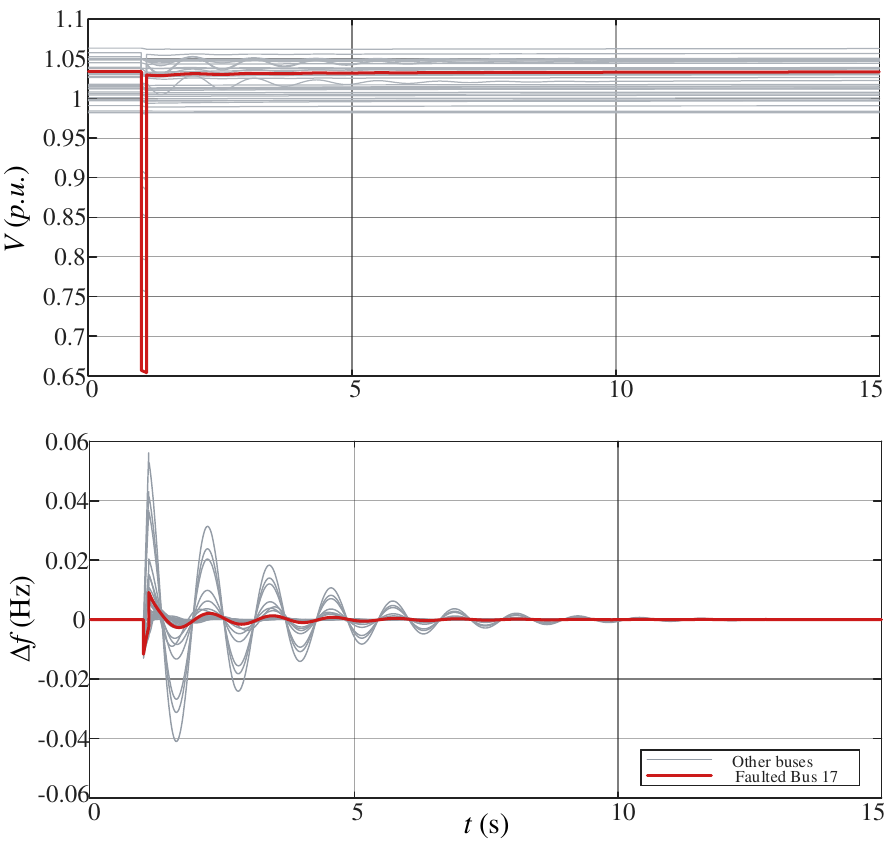}
		\caption{{Voltage and frequency responses of the modified IEEE 39-bus system to the temporary balanced three-phase fault at Bus 17.}}
		\label{fig:39bus-fault}
	\end{figure}
	}
	\subsection{Application 3: Distributed Stability Assessment and Monitoring}
	The final application highlights the compositional property of our method. Traditional approaches require centralized modeling and computation to assess the stability of interconnected power systems, which may entail heavy computational burdens and raise privacy concerns. Our method, especially Algorithm \ref{al:1}, supports a distributed cloud-edge architecture to certify the stability of large-scale systems, as shown in Fig.\ref{fig:cloud-edge}. {The procedure does not require individual subsystems to disclose detailed dynamic models or parameters to the coordinator; only the optimization messages specified in Algorithm \ref{al:1} are exchanged.}
	Here, we demonstrate this application by the IEEE 118-bus benchmark. Note that the distributed approach can be combined with the previous stability assessment to address equilibrium sets and varying operating conditions.
	\begin{figure}[htb]
		\centering
		\includegraphics[width=.85\hsize]{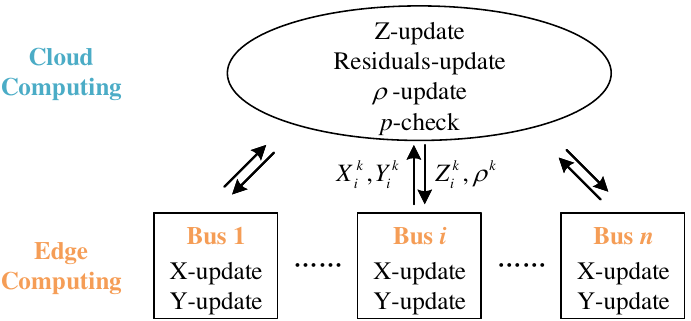}
		\caption{The distributed computing framework for stability assessment.}
		\label{fig:cloud-edge}
	\end{figure}
	\subsubsection{Case settings}
	The original IEEE 118-bus benchmark consists of 19 synchronous generators, 35 synchronous condensers, and 91 loads. To illustrate the applicability of our method to heterogeneous devices, we replace half of its SGs/condensers with inverters (CD and QD). One quarter of the loads are modeled as constant-power loads, along with the PLL dynamics, representing the aggregated dynamics of distributed renewable energy sources. The other loads are modeled as constant impedance for simplicity. Fig.\ref{fig:c5-118bus}(A) depicts the modified 118-bus power system. SG parameters follow \cite{athay1979practical}. Inverter parameters align with prior cases, with slight variations introduced to account for device heterogeneity.
	{
	\begin{figure*}[htb]
		\centering
		\includegraphics[width=\linewidth]{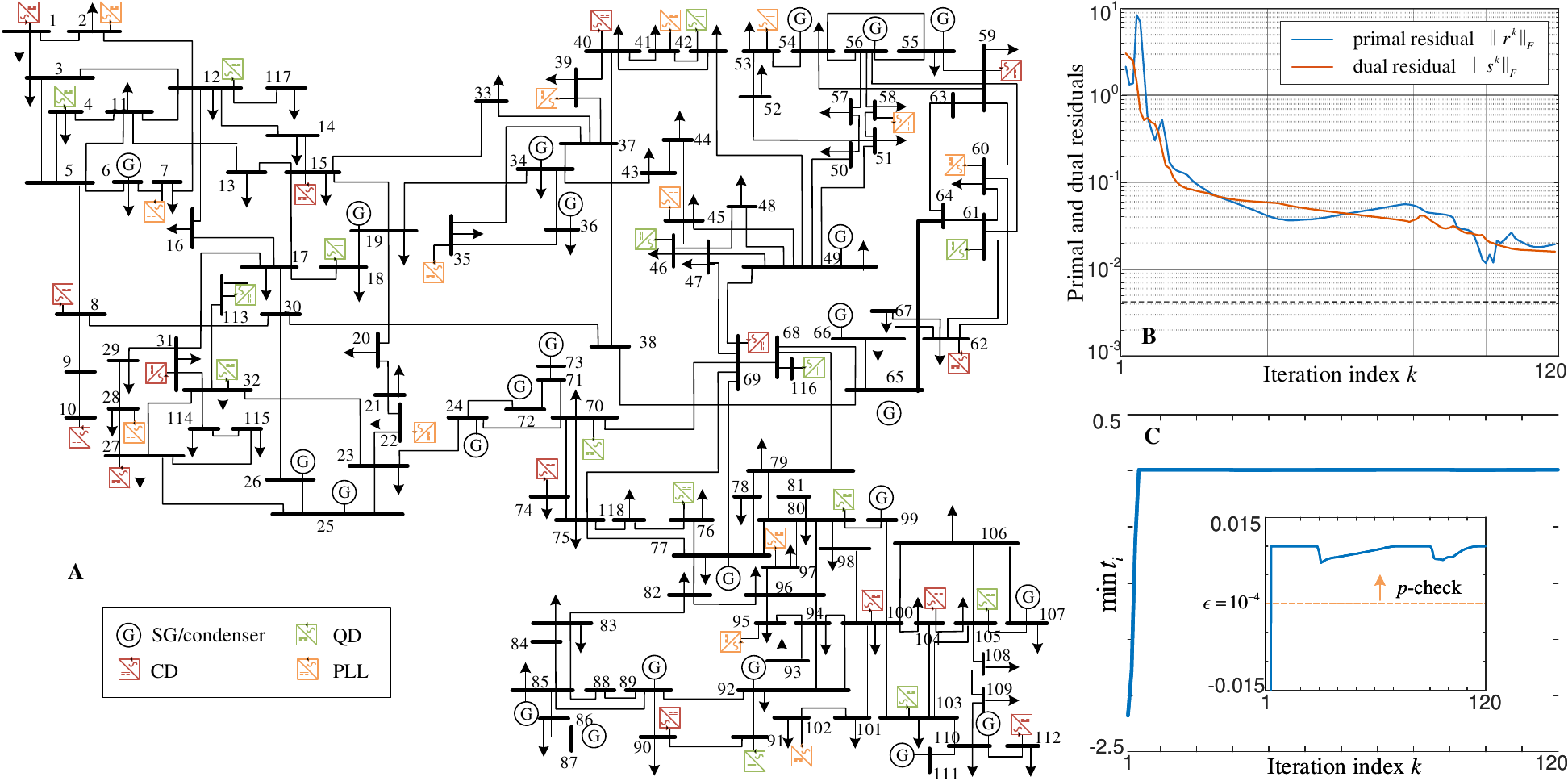}
		\caption{{(A) The modified IEEE 118-bus system. (B) The primal and dual residuals during the iteration. (C) The $p$-check is activated when all local margins $t_i$ exceed $\epsilon$ and passes at iteration 120.}}
		\label{fig:c5-118bus}
	\end{figure*}
	}
	\subsubsection{Stability assessment}
	We use Algorithm \ref{al:1} to verify the coupling condition in a distributed manner. We set $\epsilon=10^{-4}$, $\epsilon_{\mathrm{pri}}=\epsilon_{\mathrm{dual}}=0.0042$, $\mu=1.5$, $\tau^{\mathrm{incr}}=2$, $\tau^{\mathrm{decr}}=3$, $\bar{t}=10^{-2}$, and $M=1000$. {The matrices $X_i^1$, $Z_i^1$, and $Y_i^1$ are initialized to zero, while $\rho^1=1$.}
	
	{Our algorithm successfully verifies the coupling condition at iteration $k=120$, as the $p$-check passes. Fig.\ref{fig:c5-118bus}(B) shows the primal and dual residuals during the iteration, and Fig.\ref{fig:c5-118bus}(C) shows the smallest $t_i$ among all devices in each $X$-update. When $t_i>\epsilon$, the \textit{p}-check process is activated.
	In comparison, without the $p$-check process, the standard residual criterion is first satisfied at iteration 199. This shows that the $p$-check process can reduce the required iterations by 39.70\% in this case.}
	
	Algorithm \ref{al:1} yields $X_i$ for each device that satisfies LDD and the coupling condition. Hence, our theorem ensures that the equilibrium set in the LDD region $\mathcal{D}=\mathcal{D}_1\times\cdots\times\mathcal{D}_{118}$ is asymptotically stable. 
	% The nominal equilibrium was certified stable through local $\mathcal{D}_i$ checks, corroborated by centralized eigenvalue analysis, which shows that the maximal real part of eigenvalues is -0.0579.
	For verification, we consider the nominal equilibrium of the benchmark under the nominal operating condition provided by the MATPOWER toolbox\cite{zimmerman2010matpower}. Each device can independently verify that the local component of the equilibrium, i.e., $(x_i^*,u_i^*)$, lies inside $\mathcal{D}_i$. Therefore, our method directly establishes stability of the system-wide equilibrium $(x^*,u^*)$ without centralized computation. We also perform eigenvalue analysis of this equilibrium, which yields the largest real part of $-0.0579$, supporting our conclusion.
	
	\subsubsection{{Computation Cost}}
	{
	The computation was performed on a PC with Intel Core i5-10400F CPU (2.90~GHz, 6 cores and 12 threads) using MATLAB R2024a. The system has 188 state variables and 236 algebraic variables. The local LDD matrix orders range from 2 to 6, whereas the coordinator coupling LMI has order 216. For comparison, the generalized DAE pencil used for full-spectrum eigenvalue analysis has order $188+236=424$.
	The implementation of Algorithm \ref{al:1} was serial in one PC. The measured complete Algorithm \ref{al:1} wall time was 1141.96~s, comprising 145.23~s for the local $X$-updates, 790.28~s for the $Z$-updates, and 205.90~s for the $p$-checks, plus minor iteration and checkpoint overhead. Using an idealized balanced-load estimate for parallel execution of the 108 independent local $X$-updates, their cumulative time can be replaced by the average local time, giving an estimated certificate-construction time of $1141.96-145.23+145.23/108=998.08$~s. For comparison, constructing the system Jacobian and computing its eigenvalues at one known equilibrium takes approximately 147~s on the same PC. Critically, this computation must be repeated for each equilibrium because it depends on the operating point. In contrast, our method certifies stability for the entire equilibrium set within a predefined LDD region, eliminating redundant recalculations as equilibria shift. Solving for the LDD matrices $X_i$ via Algorithm \ref{al:1} incurs an initial high cost that depends on ADMM convergence (998.08 seconds in this case). Fortunately, the coupling conditions for $X_i$ are independent of the equilibrium. Once obtained, the $X_i$ matrices can be reused for stability analysis across other equilibria under different operating conditions. To illustrate, we scale the loads randomly in $(0.95,1.05)$, yielding a new operating equilibrium. With the obtained $X_i$, verifying local stability conditions takes approximately 0.04~s to certify system-wide stability. Thus, the reusable certificate becomes computationally advantageous from 7 operating points in this implementation.

	Each nonterminal communication round sends, per active subsystem, the full $4\times4$ matrix $X_i$ and scalar $t_i$ to the coordinator (17 doubles) and returns $Z_i$, $Y_i$, and $\rho$ (33 doubles). Thus, a complete round for 108 active subsystems carries 5400 doubles, or 42.19~KiB. Accounting for the final successful $p$-check, the 120-iteration run comprises 119 complete nonterminal rounds plus the final uplink and exchanges 644,436 doubles, or 4.92~MiB, over 239 one-way payload phases. This is the logical uncompressed payload; network latency and parallel runtime were not measured because the experiment used a single machine.

	}
	
	{Furthermore, the distributed architecture enhances privacy because grid-connected devices need not disclose detailed models or parameters to the coordinator, making the method suitable for multi-stakeholder grids with confidentiality requirements.}
	
	\subsubsection{{Sensitivity to ADMM Parameters}}
	{We conducted a grouped one-factor-at-a-time study to examine the penalty-update parameters and the local-margin bound. The tested values were $\mu\in\{1.2,1.5,2,5,10\}$, $\tau^{\mathrm{incr}},\tau^{\mathrm{decr}}\in\{1.5,2,3\}$, and $\bar t\in\{0.005,0.01,0.02\}$. Within each group, only the parameter under investigation was varied; the model, initialization, numerical solution configuration, and stopping tolerances were unchanged. Table \ref{tab:118bus-admm-sensitivity} reports the results.}
	
	\begin{table}[htb]
		\centering
		\scriptsize
		\caption{{Sensitivity of certification iteration counts to ADMM parameters.}}
		\label{tab:118bus-admm-sensitivity}

		\begin{tabular}{cccccc}
			\toprule
			Case & $\mu$ & $\tau^{\mathrm{incr}}$ & $\tau^{\mathrm{decr}}$ & $\bar t$ & Iteration\\
			\midrule
			1  & 10  & 2   & 2   & 0.010 & 258\\
			2  & 1.2 & 2   & 2   & 0.010 & 155\\
			3  & 1.5 & 2   & 2   & 0.010 & 168\\
			4  & 2   & 2   & 2   & 0.010 & 185\\
			5  & 5   & 2   & 2   & 0.010 & 258\\
			6  & 1.5 & 1.5 & 2   & 0.010 & 159\\
			7  & 1.5 & 3   & 2   & 0.010 & 144\\
			8  & 1.5 & 2   & 1.5 & 0.010 & 161\\
			\textbf{9} & \textbf{1.5} & \textbf{2} & \textbf{3} & \textbf{0.010} & \textbf{120}\\
			10 & 10  & 2   & 2   & 0.005 & 144\\
			11 & 10  & 2   & 2   & 0.020 & 534\\
			\bottomrule
		\end{tabular}
	\end{table}
	
	{All 11 combinations passed the $p$-check, with termination between 120 and 534 iterations. With $\tau^{\mathrm{incr}}=\tau^{\mathrm{decr}}=2$ and $\bar t=0.01$, changing $\mu$ from 10 to 1.2 reduced the iteration from 258 to 155 because the penalty update became more responsive. $\mu=5$ and $\mu=10$ triggered no penalty update and both required 258 iterations. The effects of the two update factors were nonmonotonic over the tested values. With $\mu=10$ and both factors equal to 2, increasing $\bar t$ from 0.005 to 0.01 and 0.02 increased the iteration from 144 to 258 and 534, respectively. The selected setting $(\mu,\tau^{\mathrm{incr}},\tau^{\mathrm{decr}},\bar t)=(1.5,2,3,0.01)$ required the fewest iterations among the tested combinations. These results provide practical tuning guidance for the 118-bus case.}

	{Taken together, the three applications illustrate the operational value of equilibrium-set-oriented certification. In the 9-bus case, one common certificate establishes the stability of two coexisting isolated equilibria contained in the certified region, without constructing separate point-specific certificates. In the 39-bus case, the matrices $X_i$ and the network coupling certificate are reused while the local LDD conditions are verified as the load factor shifts the operating equilibrium, yielding a certified operating range $s\in[0.8535,1.1625]$. The 118-bus case demonstrates how the same certificate-construction and local-screening workflow can be implemented in a large heterogeneous system through distributed computation. In practice, once the regional certificates and network coupling condition have been established, a candidate equilibrium can be assessed through local region and validity checks without reconstructing an equilibrium-specific compositional certificate or the complete system Jacobian. This reuse across compatible equilibria and operating conditions shows the principal practical value of the equilibrium-set result.}

	{Beyond the isolated equilibria considered in the present case studies, the equilibrium-set formulation is also relevant to nonisolated equilibrium continua. A power-system example in \cite{10239448} exhibits a continuous equilibrium trajectory to which different initial conditions may converge at different points, although no individual equilibrium on the continuum is asymptotically stable. This situation illustrates the distinction between asymptotic stability of an equilibrium set and pointwise asymptotic stability of an isolated equilibrium.}
	
	\section{Concluding Remarks}
	Building on the theoretical framework developed in Part I, this paper focuses on methods for applying the {compositional} and equilibrium-free approach to power systems. We addressed two critical challenges: verifying local delta dissipativity for heterogeneous devices and efficiently certifying coupling conditions in interconnected systems.
	We proposed a method using Krasovskii-type storage functions to verify local conditions, adaptable to a wide range of nonlinear power device models. For coupling conditions, we developed an ADMM-based distributed algorithm that ensures scalability for large-scale power grids.
	Three key applications of our theory were demonstrated, highlighting the advantages of its equilibrium-free and {compositional} features. Case studies on modified IEEE 9-bus, 39-bus, and 118-bus benchmarks validated our methods, showcasing their effectiveness across systems of varying scales and complexities.
	The present study also has several limitations. First, the proposed conditions are sufficient but not necessary, so the certification may be conservative for some operating regions. Second, the certified LDD regions depend on model structure and the selected storage-function parameterization. Third, although distributed verification improves scalability and privacy, convergence speed remains influenced by ADMM parameter tuning and network size. These limitations motivate future work on tighter certificates, adaptive parameter strategies, and broader model coverage.
	
	This two-part study presents a comprehensive framework for scalable, modular stability analysis in modern power systems. By eliminating reliance on a single fixed equilibrium point and leveraging distributed algorithms, the approach addresses key stability-analysis challenges in power grids with massive, heterogeneous devices and continuously varying operating conditions. Future work may integrate local delta dissipativity with advanced grid-control strategies to further enhance power system stability. {Extending the present continuous-time framework to hybrid models and sampled-data implementations are also important directions for future work.}
	\bibliographystyle{IEEEtran}
	\bibliography{mybib}

@article{huang2021decentralized,
  title={Decentralized data-enabled predictive control for power system oscillation damping},
  author={Huang, Linbin and Coulson, Jeremy and Lygeros, John and D{\"o}rfler, Florian},
  journal={IEEE Trans. Control Syst. Technol.},
  volume={30},
  number={3},
  pages={1065--1077},
  year={2022},
  publisher={IEEE}
}

@article{mevsanovic2021scalable,
  title={Scalable and data privacy conserving controller tuning for large-scale power networks},
  author={Me{\v{s}}anovi{\'c}, Amer and M{\"u}nz, Ulrich and Findeisen, Rolf},
  journal={IEEE Trans. Control Syst. Technol.},
  volume={30},
  number={2},
  pages={696--711},
  year={2022},
  publisher={IEEE}
}

@article{esfahani2014performance,
	title={Performance bounds for the scenario approach and an extension to a class of non-convex programs},
	author={Esfahani, Peyman Mohajerin and Sutter, Tobias and Lygeros, John},
	journal={IEEE Transactions on Automatic Control},
	volume={60},
	number={1},
	pages={46--58},
	year={2015},
	publisher={IEEE}
}

@article{YP_part-I,
	title={{{Compositional and Equilibrium-Free Stability Certification for Power Systems--Part I: Theory}}},
	author={Yang, Peng and Peng, Xiaoyu and Ru, Xi and Geng, Hua and Liu, Feng},
	journal={{Submitted to IEEE Trans. Control Syst. Technol.}},
	year={2026}
}

@article{ben2002robust,
	title={Robust optimization--methodology and applications},
	author={Ben-Tal, Aharon and Nemirovski, Arkadi},
	journal={Math. Program.},
	volume={92},
	pages={453--480},
	year={2002},
	publisher={Springer}
}

@inproceedings{topcu2009compositional,
	title={Compositional stability analysis based on dual decomposition},
	author={Topcu, Ufuk and Packard, Andrew K and Murray, Richard M},
	booktitle={Proceedings of the 48th IEEE Conference on Decision and Control (CDC) held jointly with the 28th Chinese Control Conference},
	pages={1175--1180},
	year={2009},
	organization={IEEE},
	doi={10.1109/CDC.2009.5400309}
}

@article{meissen2015compositional,
	title={Compositional performance certification of interconnected systems using ADMM},
	author={Meissen, Chris and Lessard, Laurent and Arcak, Murat and Packard, Andrew K},
	journal={Automatica},
	volume={61},
	pages={55--63},
	year={2015},
	doi={10.1016/j.automatica.2015.07.027},
	publisher={Elsevier}
}

@article{kottenstette2014relationships,
	title={On relationships among passivity, positive realness, and dissipativity in linear systems},
	author={Kottenstette, Nicholas and McCourt, Michael J and Xia, Meng and Gupta, Vijay and Antsaklis, Panos J},
	journal={Automatica},
	volume={50},
	number={4},
	pages={1003--1016},
	year={2014},
	publisher={Elsevier}
}

@article{hill1976stability,
	title={The stability of nonlinear dissipative systems},
	author={Hill, David and Moylan, Peter},
	journal={IEEE Trans. Autom. Control},
	volume={21},
	number={5},
	pages={708--711},
	year={1976},
	publisher={IEEE}
}

@article{gusev2006kalman,
	title={Kalman-Popov-Yakubovich lemma and the S-procedure: A historical essay},
	author={Gusev, Sergei V and Likhtarnikov, Andrey Leonidovich},
	journal={Autom. Remote Control},
	volume={67},
	pages={1768--1810},
	year={2006},
	publisher={Springer}
}

@article{boyd2011distributed,
	title={Distributed optimization and statistical learning via the alternating direction method of multipliers},
	author={Boyd, Stephen and Parikh, Neal and Chu, Eric and Peleato, Borja and Eckstein, Jonathan and others},
	journal={Found. Trends Mach. Learn.	},
	volume={3},
	number={1},
	pages={1--122},
	year={2011},
	publisher={Now Publishers, Inc.}
}

@article{molzahn2017survey,
	title={A survey of distributed optimization and control algorithms for electric power systems},
	author={Molzahn, Daniel K and D{\"o}rfler, Florian and Sandberg, Henrik and Low, Steven H and Chakrabarti, Sambuddha and Baldick, Ross and Lavaei, Javad},
	journal={IEEE Trans. Smart Grid},
	volume={8},
	number={6},
	pages={2941--2962},
	year={2017},
	publisher={IEEE}
}

@book{anderson2008power,
	title={Power System Control and Stability},
	edition={2},
	author={Anderson, P. M. and Fouad, A. A.},
	isbn={978-0-471-23862-1},
	year={2003},
	address={Piscataway, NJ, USA},
	publisher={IEEE Press/Wiley-Interscience}
}

@book{Pai_EnergyFunctionAnalysis_1989,
	title = {Energy Function Analysis for Power System Stability},
	author = {Pai, M. A.},
	year = {1989},
	address={Boston, MA, USA},
	publisher = {Kluwer Academic Publishers},
	isbn = {978-0-7923-9035-0},
	doi={10.1007/978-1-4613-1635-0}
}

@article{bevrani2014virtual,
	title={Virtual synchronous generators: a survey and new perspectives},
	author={Bevrani, Hassan and Ise, Toshifumi and Miura, Yushi},
	journal={Int. J. Electr. Power Energy Syst.},
	volume={54},
	pages={244--254},
	year={2014},
	publisher={Elsevier}
}

@article{huang2022impacts,
	title={Impacts of grid structure on PLL-synchronization stability of converter-integrated power systems},
	author={Huang, Linbin and Xin, Huanhai and Dong, Wei and D{\"o}rfler, Florian},
	journal={IFAC-PapersOnLine},
	volume={55},
	number={13},
	pages={264--269},
	year={2022},
	publisher={Elsevier}
}

@ARTICLE{10239448,
	author={Yang, Peng and Liu, Feng and Liu, Tao and Hill, David J.},
	journal={IEEE Trans. Autom. Control}, 
	title={Augmented Synchronization of Power Systems}, 
	year={2024},
	volume={69},
	number={6},
	pages={3673-3688},
	doi={10.1109/TAC.2023.3312124}}

@ARTICLE{8890862,
	author={P. {Yang} and F. {Liu} and Z. {Wang} and C. {Shen}},
	journal={IEEE Trans. Power Syst.}, 
	title={Distributed Stability Conditions for Power Systems With Heterogeneous Nonlinear Bus Dynamics}, 
	year={2020},
	volume={35},
	number={3},
	pages={2313-2324},
	doi={10.1109/TPWRS.2019.2951202}}

@article{zimmerman2010matpower,
	title={MATPOWER: Steady-state operations, planning, and analysis tools for power systems research and education},
	author={Zimmerman, Ray Daniel and Murillo-S{\'a}nchez, Carlos Edmundo and Thomas, Robert John},
	journal={IEEE Transactions on power systems},
	volume={26},
	number={1},
	pages={12--19},
	year={2011},
	publisher={IEEE}
}

@article{simpson2018equilibrium,
	title={Equilibrium-independent dissipativity with quadratic supply rates},
	author={Simpson-Porco, John W},
	journal={IEEE Trans. Autom. Control},
	volume={64},
	number={4},
	pages={1440--1455},
	year={2019},
	publisher={IEEE}
}

@article{athay1979practical,
	title={A practical method for the direct analysis of transient stability},
	author={Athay, T and Podmore, R and Virmani, S},
	journal={IEEE Trans. Power Apparatus and Systems},
	number={2},
	pages={573--584},
	year={1979},
	publisher={IEEE}
}

@article{Simpson-Porco_VoltageStabilizationMicrogrids_2017,
	title={Voltage stabilization in microgrids via quadratic droop control},
	author={Simpson-Porco, John W and D{\"o}rfler, Florian and Bullo, Francesco},
	journal={IEEE Trans. Autom. Control},
	volume={62},
	number={3},
	pages={1239--1253},
	year={2017},
	publisher={IEEE}
}

@article{Stegink_unifyingenergybasedapproach_2016,
	title={A unifying energy-based approach to stability of power grids with market dynamics},
	author={Stegink, Tjerk and De Persis, Claudio and van der Schaft, Arjan},
	journal = {IEEE Trans. Autom. Control},
	volume={62},
	number={6},
	pages={2612--2622},
	year={2017},
	publisher={IEEE}
}

@book{Kundur_Powersystemstability_1994,
	title={Power system stability and control},
	author={Kundur, Prabha S.},
	year={1994},
	address={New York, NY, USA},
	publisher={McGraw-Hill}
}

@article{Zhang_OnlineDynamicSecurity_2015,
	title={Online dynamic security assessment of microgrid interconnections in smart distribution systems},
	author={Zhang, Yun and Xie, Le},
	journal = {IEEE Trans. Power Syst.},
	volume={30},
	number={6},
	pages={3246--3254},
	year={2015},
	publisher={IEEE}
}

@article{Stott_Powersystemdynamic_1979,
	title={Power system dynamic response calculations},
	author={Stott, Brian},
	journal={Proc. IEEE},
	volume={67},
	number={2},
	pages={219--241},
	year={1979},
	publisher={IEEE}
}

@article{Sastry_Hierarchicalstabilityalert_1980,
	title={Hierarchical stability and alert state steering control of interconnected power systems},
	author={Sastry, Shankar and Varaiya, Pravin},
	journal={IEEE Trans. Circuits Syst.},
	volume={27},
	number={11},
	pages={1102--1112},
	year={1980},
	publisher={IEEE}
}

@article{Fiaz_portHamiltonianapproachpower_2013a,
	title={A port-Hamiltonian approach to power network modeling and analysis},
	author={Fiaz, Shaik and Zonetti, Daniele and Ortega, Romeo and Scherpen, Jacquelien MA and Van der Schaft, AJ},
	journal={Eur. J. Control},
	volume={19},
	number={6},
	pages={477--485},
	year={2013},
	publisher={Elsevier}
}

@article{Chiang_Directstabilityanalysis_1995,
	title={Direct stability analysis of electric power systems using energy functions: theory, applications, and perspective},
	author={Chiang, Hsiao-Dong and Chu, Chia-Chi and Cauley, Gerry},
	journal={Proc. IEEE},
	volume={83},
	number={11},
	pages={1497--1529},
	year={1995},
	publisher={IEEE}
}

@article{xia2021dissipativity,
	author={J. Xia and G. Chen and J. H. Park and H. Shen and G. Zhuang},
	title={Dissipativity-Based Sampled-Data Control for Fuzzy Switched Markovian Jump Systems},
	journal={IEEE Trans. Fuzzy Syst.},
	volume={29},
	number={6},
	pages={1325--1339},
	year={2021},
	doi={10.1109/TFUZZ.2020.2970856}}

@article{chen2020passivity,
	author={G. Chen and J. Sun and J. Chen},
	title={Passivity-Based Robust Sampled-Data Control for Markovian Jump Systems},
	journal={IEEE Trans. Syst., Man, Cybern., Syst.},
	volume={50},
	number={7},
	pages={2671--2684},
	year={2020},
	doi={10.1109/TSMC.2018.2825474}}

@inproceedings{ahmadi2018distributed,
	author={M. Ahmadi and M. Cubuktepe and U. Topcu and T. Tanaka},
	title={Distributed Synthesis Using Accelerated {ADMM}},
	booktitle={Proc. Amer. Control Conf. (ACC)},
	pages={6206--6211},
	year={2018},
	doi={10.23919/ACC.2018.8431859}}

@article{chen2025asynchronous,
	author={G. Chen and X. Liu and X. Ren and J. Xia and J. H. Park},
	title={Asynchronous-Based Countermeasure for Stealthy Attack on Aperiodic Sampled-Data Control Systems},
	journal={IEEE Trans. Autom. Control},
	volume={70},
	number={9},
	pages={6159--6166},
	year={2025},
	doi={10.1109/TAC.2025.3552745}}

@article{yang2026nonuniform,
	author={T. Yang and M. Tang and G. Chen and J. Xia and J. H. Park},
	title={Design of Non-Uniform Sampled-Data Load Frequency Control Scheme for Multi-Area Wind Power Integrated Power System and Its Application in {IEEE} 39-Bus System},
	journal={IEEE Trans. Consum. Electron.},
	volume={72},
	number={2},
	pages={4029--4039},
	year={2026},
	doi={10.1109/TCE.2026.3666872}}

@article{tang2024multiple,
	author={M. Tang and G. Chen and J. Xia and X. Xie and J. H. Park},
	title={Load Frequency Control of Power Systems With Multiple Transmission Delays Under Aperiodic Sampled-Data},
	journal={IEEE Trans. Circuits Syst. I, Reg. Papers},
	volume={71},
	number={12},
	pages={5884--5895},
	year={2024},
	doi={10.1109/TCSI.2024.3406503}}

@article{yang2026aperiodic,
	author={T. Yang and X. Han and G. Chen and J. Xia and J. H. Park},
	title={Aperiodic Sampling Load Frequency Control of Wind Power-Integrated Power Systems With Packet Losses and Delays Under Discrete Time Feedback Interconnection Model},
	journal={IEEE Trans. Smart Grid},
	volume={17},
	number={1},
	pages={858--869},
	year={2026},
	doi={10.1109/TSG.2025.3614697}}

@article{vandenberghe1996semidefinite,
	author={L. Vandenberghe and S. Boyd},
	title={Semidefinite Programming},
	journal={SIAM Rev.},
	volume={38},
	number={1},
	pages={49--95},
	year={1996},
	doi={10.1137/1038003}}

@article{moler1973algorithm,
	author={C. B. Moler and G. W. Stewart},
	title={An Algorithm for Generalized Matrix Eigenvalue Problems},
	journal={SIAM J. Numer. Anal.},
	volume={10},
	number={2},
	pages={241--256},
	year={1973},
	doi={10.1137/0710024}}

\end{document}